\documentclass[10pt]{iopart}

\usepackage{iopams}
\usepackage{amssymb}
\usepackage{amsfonts}
\usepackage{amstext}
\usepackage{graphicx}
\usepackage{setstack}
\usepackage{amscd}
\usepackage{color}

\newcommand{\ihbar}{\imath \hbar}
\renewcommand{\S}{{\mathcal S}}
\newcommand{\E}{{\mathcal E}}

\newcommand{\Ted}{\underset{\rightarrow}{\mathbb{T}e}}
\newcommand{\Teg}{\underset{\leftarrow}{\mathbb{T}e}}
\newcommand{\llangle}{\langle \hspace{-0.2em} \langle}
\newcommand{\rrangle}{\rangle \hspace{-0.2em} \rangle}
\newcommand{\dist}{\mathrm{dist}}
\newcommand{\Ad}{\mathfrak{Ad}}
\newcommand{\id}{\mathrm{id}}

\newtheorem{theo}{Theorem}
\newtheorem{propo}{Proposition}
\newtheorem{cor}{Corollary}

\newenvironment{proof}{\noindent \textit{Proof:}}{\hfill $\Box$ \\}

\begin{document}

\title{Adiabatic theorem for bipartite quantum systems in weak coupling limit}

\author{David Viennot \& Lucile Aubourg}
\address{Institut UTINAM (CNRS UMR 6213, Universit\'e de Franche-Comt\'e), 41bis Avenue de l'Observatoire, BP1615, 25010 Besan\c con cedex, France.}

\begin{abstract}
We study the adiabatic approximation of the dynamics of a bipartite quantum system with respect to one of the components, when the coupling between its two components is perturbative. We show that the density matrix of the considered component is described by adiabatic transport formulae exhibiting operator-valued geometric and dynamical phases. The present results can be used to study the quantum control of the dynamics of qubits and of open quantum systems where the two components are the system and its environment. We treat two examples, the control of an atomic qubit interacting with another one and the control of a spin in the middle of a Heisenberg spin chain.
\end{abstract}

\pacs{03.65.Vf, 02.30.Mv, 03.65.Ud}

\section{Introduction}
A bipartite quantum system consists of two quantum subsystems denoted by $\S$ and $\E$ and described by the Hilbert space $\mathcal H_\S \otimes \mathcal H_\E$. We are interested only by the behaviour of the component $\S$ with a loss of the informations concerning the component $\E$. If $\psi \in \mathcal H_\S \otimes \mathcal H_\E$ is the state of the bipartite system, the subsystem $\S$ is described by the density matrix (the mixed state) $\rho = \tr_\E |\psi \rrangle \llangle \psi|$ ($\llangle .|. \rrangle$ denotes the inner product of $\mathcal H_\S \otimes \mathcal H_\E$) where the partial trace $\tr_\E$ on $\mathcal H_\E$ suppresses the informations concerning $\E$. $\rho$ takes into account the entanglement of $\S$ with $\E$. An interesting problem in the dynamics of a bipartite quantum system is the control of the component $\S$ ``hampered'' by $\E$. Quantum control has potentially a lot of applications in nanoscience and in quantum computing. A problem of quantum control consists to find how acting on $\S$ (by laser fields, magnetic fields, etc) in order to $\S$ evolves from its initial state $\rho_0$ to a predetermined target state $\rho_{target}$ (the goal of the control). The presence of $\E$ can considerably modify the control problem with regard to the control of a pure state in $\mathcal H_\S$ for $\S$ alone. Moreover the control can affect $\E$ directly or indirectly by the coupling between $\E$ and $\S$. The understanding of the dynamics of $\S$ in contact with $\E$ is crucial.\\
Such a situation occurs in quantum information theory where $\S$ and $\E$ are two qubits of a quantum computer or two ensembles of qubits. In this case the control consists to perform a logical gate or a quantum algorithm on $\S$ in the presence of other qubits ($\E$) of the quantum computer. In this example, $\S$ and $\E$ have similar sizes and the role of $\S$ and $\E$ can be exchanged. But dynamics of bipartite quantum systems also occur in the control of open quantum systems where $\S$ is a small subsystem and $\E$ is a large environment responsible for decoherence effects on $\S$. In this case the loss of information models the observer unknowledge concerning $\E$ due to its large size and its complexity.\\
An adiabatic approach \cite{Messiah} can be used to describe quantum control since the variations of the control parameters are often slow. Quantum control schemes based on adiabatic approximation have been proposed for different closed systems \cite{Zanardi,Lucarelli,Santoro,Guerin}. For open quantum systems, adiabatic approaches based on non-hermitian Hamiltonians have been studied \cite{Fleischer,Sarandy1,Sarandy2}. For example, the Lindblad equation (a Markovian approximation of the dynamics of $\S$ in the environment \cite{Breuer}) is considered as a non-hermitian Schr\"odinger equation in the Hilbert-Schmidt space (the so-called Liouville space describing the square trace class operators of $\mathcal H_\S$, i.e. $\tr_\S(A^\dagger A) < \infty$). In these case the adiabatic approximation is based on adiabatic theorems for non-self-adjoint Hamiltonians \cite{Nenciu1,Salem,Joye}. But these works do not focus on bipartite aspect of quantum dynamics, but which has been studied by Sj\"oqvist {\it etal} with the viewpoint of non-adiabatic geometric phases \cite{Sjoqvist,Tong,Dajka}. Recently operator-valued geometric phases have been proposed as generalizations of the Sj\"oqvist geometric phases, in the context of cyclic (non-adiabatic) evolutions \cite{Andersson} and of adiabatic evolutions \cite{Viennot1,Viennot2}. Nevertheless a rigorous study of the adiabatic regime of a bipartite quantum system exhibiting operator valued geometric phases has never been realized.\\
The goal of this paper is to show that under common assumptions concerning the evolution of the bipartite system, the evolution of the density matrix satisfies adiabatic transport formulae exhibiting operator-valued geometric (and dynamical) phases. Section II establishes adiabatic theorems with regard to a discussion concerning the different time scales involved in the dynamics of a bipartite quantum system. In particular we consider two adiabatic regimes. These results are based on the Nenciu adiabatic theorem \cite{Nenciu2} which considers spectral components and not only one eigenvalue. Adiabatic transport formulae for the density matrix are obtained in section III. We show that these formulae (for the weak adiabatic regime) exhibit operator valued geometric phases similar to the ones introduced in the previous works \cite{Sjoqvist,Tong,Dajka,Andersson,Viennot1,Viennot2}. The generator of the operator valued dynamical phase also exhibited by these formulae can appear as an effective Hamiltonian of $\S$ dressed by $\E$. The adiabatic transport of the density matrix for the second order perturbative expansion, satisfies a kind of effective Lindblad equation. We discuss also in this section the specific case of the open quantum systems where $\E$ is a thermal bath. Throughout this paper we consider the case of a weak coupling between $\S$ and $\E$. Indeed to enlighten the individual behaviour of the component $\S$ inside the bipartite system through the viewpoint of the adiabatic approximation, it needs to explain the relation between the eigenvectors of the bipartite system and the eigenvectors of its component $\S$. This requires a perturbative analysis. Finally section IV presents two examples: a qubit realized as a two level atom in a laser field in the rotating wave approximation with a perturbative interaction with another atom; and a spin controlled by a magnetic field in the middle of a ferromagnetic spin chain with Heisenberg coupling between the nearest neighbours.

\section{Strong and weak adiabatic theorems}
\subsection{Preliminary discussion}
\label{21}
We consider a bipartite quantum system of Hilbert space $\mathcal H_\S \otimes \mathcal H_\E$ governed by the time-dependent Hamiltonian
\begin{equation}
H(t) = H_\S(t) \otimes 1_\E + 1_\S \otimes H_\E(t) + \epsilon V(t)
\end{equation}
where $H_\S \in \mathcal L(\mathcal H_\S)$ is the selfadjoint Hamiltonian of the component $\S$, $H_\E \in \mathcal L(\mathcal H_\E)$ is the selfadjoint Hamiltonian of the component $\E$, $V \in \mathcal L(\mathcal H_\S \otimes \mathcal H_\E)$ is the coupling operator between $\S$ and $\E$, and $\epsilon \in \mathcal V(0)$ is a perturbative parameter ($\mathcal V(0)$ denotes the neighbourhood of $0$). Let $\{\mu_b\}_b$ and $\{\nu_\beta\}_\beta$ be the pure point spectra of $H_\S$ and $H_\E$ (for the sake of simplicity we suppose that all these eigenvalues are not globally denegenerate, e.g. these eigenvalues are non degenerate for all $t>0$ except possibly for a finite number of isolated moments $t_*$) and $\{\zeta_b\}_b$ and $\{\xi_\beta\}_\beta$ be the associated normalized eigenvectors. The eigenvalues are supposed at least $\mathcal C^0$ and the eigenvectors are supposed at least $\mathcal C^1$ with respect to $t$.
\begin{eqnarray}
H_\S(t) \zeta_b(t) & = & \mu_b(t) \zeta_b(t) \qquad \zeta_b \in \mathcal H_\S, \mu_b \in \mathbb R \\
H_\E(t) \xi_\beta(t) & = & \nu_\beta(t) \xi_\beta(t) \qquad \xi_\beta \in \mathcal H_\E, \nu_\beta \in \mathbb R
\end{eqnarray}
Let $\{\lambda_{b\beta}\}_{b,\beta}$ be the perturbed pure point spectrum of $H$ and $\{\phi_{b \beta}\}_{b,\beta}$ be the associated eigenvectors.
\begin{equation}
H(t) \phi_{b\beta}(t) = \lambda_{b\beta}(t) \phi_{b\beta}(t) \qquad \phi_{b\beta} \in \mathcal H_\S \otimes \mathcal H_\E, \lambda_{b \beta} \in \mathbb R
\end{equation}
\begin{eqnarray}
\lim_{\epsilon \to 0} \phi_{b \beta}(t) & = & \zeta_b(t) \otimes \xi_\beta(t) \\ \lim_{\epsilon \to 0} \lambda_{b \beta}(t) & = & \mu_b(t) + \nu_\beta(t) 
\end{eqnarray}
(with the quantum state limit defined with the norm topology associated with $\llangle .|. \rrangle$). The first order approximations (by using the Rayleigh-Schr\"odinger perturbation method) are
\begin{eqnarray}
\lambda_{b \beta} & = & \mu_b + \nu_\beta + \epsilon V_{b\beta,b\beta} + \mathcal O(\epsilon^2) \\
\phi_{b \beta} & = & \zeta_b \otimes \xi_\beta + \epsilon \sum_{(c\gamma)\not=(b\beta)} \frac{V_{c\gamma,b\beta}}{\mu_b-\mu_c+\nu_\beta-\nu_\gamma} \zeta_c \otimes \xi_\gamma + \mathcal O(\epsilon^2)
\end{eqnarray}
where $V_{c\gamma,b\beta} = \llangle \zeta_c \otimes \xi_\gamma|V \zeta_b \otimes \xi_\beta \rrangle$. We consider the dynamics of the bipartite system starting from $\phi_{a\alpha}$, i.e.
\begin{equation}
\ihbar \frac{d\psi(t)}{dt} = H(t) \psi(t) \qquad \psi(0) = \phi_{a\alpha}(0)
\end{equation}
From the viewpoint of the control of $\S$, there are three time scales:
\begin{itemize}
\item $T$ the total duration of the evolution (the duration of the control);
\item $\tau_\S$ the quantum proper time characterizing the transition of $\S$ from $\zeta_a$ to another eigenvector, induced by the control (the Rabi period of the first transition involving $\zeta_a$, e.g. $\tau_\S = \sup_{t\in[0,T]} \max_{b \not= a} \frac{\hbar}{|\mu_b - \mu_a|}$);
\item $\theta^\epsilon$ the time characterizing the perturbation of $\S$ by $\E$.
\end{itemize}
We remark that the non-selfadjoint models \cite{Nenciu1,Salem,Joye} exhibit also three time scales (the duration of the evolution, the time characterizing the quantum transitions, and the time characterizing the dissipation -- the inverse of the resonance width --).\\
There are then three adiabatic regimes:
\begin{itemize}
\item $\tau_\S \ll \theta^\epsilon$ and $\tau_\S \ll T$ (very strong adiabatic regime);
\item $\tau_\S \sim \theta^\epsilon$ and $\theta^\epsilon \ll T$ (strong adiabatic regime);
\item $\tau_\S \sim \theta^\epsilon$ and $\theta^\epsilon \sim T$, the evolution of $\E$ being assumed to be very strongly adiabatic (weak adiabatic regime).
\end{itemize}
More precisely, consider the non-adiabatic couplings (for $(b\beta) \not= (a\alpha)$):
\begin{equation}
\llangle \phi_{a\alpha} | \dot \phi_{b\beta} \rrangle = \frac{\hbar}{T(\lambda_{b\beta} - \lambda_{a\alpha})} \llangle \phi_{a\alpha} | \hbar^{-1} H'| \phi_{b \beta} \rrangle
\end{equation}
where a dot denotes the derivative with respect to $t$ and a prime denotes the derivative with respect to the reduced time $s = t/T$. If $\forall b\not=a$, $\inf_{t \in [0,T]} |\mu_b - \mu_a| = \mathcal O(1)$ ($\iff \tau_\S \ll \theta^\epsilon$, $\mathcal O(1)$ means a gap condition very large with respect to $\epsilon$ [a value of zero order in $\epsilon$]) then
\begin{eqnarray}
\frac{\hbar}{T(\lambda_{b\beta} - \lambda_{a\alpha})} & = & \frac{\hbar}{T (\mu_b-\mu_a) \left(1 + \frac{\nu_\beta-\nu_\alpha}{\mu_b-\mu_a} \right)} \nonumber \\
& & \quad - \frac{\hbar \epsilon (V_{b\beta,b\beta}-V_{a\alpha,a\alpha})}{T(\mu_b-\mu_a)^2 \left(1 + \frac{\nu_\beta-\nu_\alpha}{\mu_b-\mu_a} \right)^2} + \mathcal O(\epsilon^2)
\end{eqnarray}
By assuming that $\Delta = \inf_{t\in[0,T]} \min_{\beta\not=\alpha} \max_{b \not=a} \left|1 + \frac{\nu_\beta-\nu_\alpha}{\mu_b-\mu_a} \right| = \mathcal O(1)$ (no resonance between a transition of $\S$ from $\zeta_a$ and a transition of $\E$ from $\xi_\alpha$) we have
\begin{equation}
\left|\frac{\hbar}{T(\lambda_{b\beta} - \lambda_{a\alpha})} \right| \leq \frac{\tau_\S}{T \Delta} + \frac{\tau_\S^2}{T \theta^\epsilon \Delta^2} + \mathcal O(\epsilon^2)
\end{equation} 
with $\tau_S = \sup_{t \in [0,1]} \max_{b \not=a} \frac{\hbar}{|\mu_b-\mu_a|}$ and $\theta^\epsilon = \inf_{t \in [0,T]} \min_{(b\beta) \not= (a\alpha)} \frac{\hbar}{\epsilon|V_{b\beta,b\beta}-V_{a\alpha,a\alpha}|}$. If $T$ is chosen like $\tau_\S \ll T$, then $\forall (b\beta) \not= (a\alpha)$, $|\llangle \phi_{a\alpha} | \dot \phi_{b\beta} \rrangle| \ll 1$. All non-adiabatic couplings being negligible, we can think that the system remains projected only onto $\phi_{a\alpha}(t)$ during the whole dynamics. This is the very strong adiabatic regime which corresponds to an adiabatic evolution of the whole bipartite system.\\
Now if $\inf_{t\in[0,T]} \min_{b \not=a} |\mu_b - \mu_a| = \mathcal O(\epsilon)$ ($\iff \tau_\S \sim \theta^\epsilon$) we have
\begin{itemize}
\item if $\alpha \not= \beta$ and $\inf_{t \in [0,T]} |\nu_\beta-\nu_\alpha| = \mathcal O(1)$:
\begin{eqnarray}
\frac{\hbar}{T(\lambda_{b\beta}-\lambda_{a\alpha})} & = & \frac{\hbar}{T(\nu_\beta-\nu_\alpha)} \nonumber \\
& & \quad - \frac{\hbar \epsilon(\tilde \mu_b - \tilde \mu_a + V_{b\beta,b\beta}-V_{a\alpha,a\alpha})}{T(\nu_\beta - \nu_\alpha)^2} + \mathcal O(\epsilon^2)
\end{eqnarray}
where $\tilde \mu_b = \frac{\mu_b}{\epsilon}$ ($\inf_{t\in[0,T]} \min_{b \not=a} |\tilde \mu_b - \tilde \mu_a| = \mathcal O(1)$).
\item if $\alpha = \beta$ or $\exists t_*$ such that $\nu_\beta(t_*) = \nu_\alpha(t_*)$: since $\lambda_{b\beta} = \epsilon (\tilde \mu_b + V_{b\beta,b\beta}) + \nu_\beta$ we have
\begin{equation}
\frac{\hbar}{T(\lambda_{b\beta}-\lambda_{a\alpha})} = \frac{\hbar}{T\epsilon (\tilde \mu_b - \tilde \mu_a + V_{b\beta,b\beta} - V_{a\alpha,a\alpha})}
\end{equation}
for all $t$ if $\alpha=\beta$ or only at $t=t_*$.
\end{itemize}
 We have
\begin{itemize}
\item if $\alpha \not= \beta$ and $\inf_{t \in [0,T]} |\nu_\beta-\nu_\alpha| = \mathcal O(1)$:
\begin{equation}
\left|\frac{\hbar}{T(\lambda_{b\beta}-\lambda_{a\alpha})} \right| \leq \frac{\tau_\E}{T} + \frac{\tau_\E^2}{T \theta^\epsilon} + \mathcal O(\epsilon^2)
\end{equation}
\item if $\alpha = \beta$ or $\exists t_*$ such that $\nu_\beta(t_*) = \nu_\alpha(t_*)$:
\begin{equation}
\left| \frac{\hbar}{T(\lambda_{b\beta}-\lambda_{a\alpha})} \right| \leq \frac{\tau_\S^\epsilon}{T}
\end{equation}
for all $t$ if $\alpha=\beta$ or only at $t=t_*$. 
\end{itemize}
with $\tau_\E = \sup_{t\in[0,T]} \max_{\beta\not=\alpha} \frac{\hbar}{|\nu_\beta-\nu_\alpha|}$, $\theta^\epsilon = \inf_{t \in [0,T]} \min_{(b\beta)\not=(a\alpha)} \frac{\hbar}{\epsilon(\tilde \mu_b - \tilde \mu_a + V_{b\beta,b\beta} - V_{a\alpha,a\alpha})}$ and $\tau_\S^\epsilon = \sup_{t \in [0,T]} \max_{\beta s.t. \nu_\beta(t_*)= \nu_\alpha(t_*)} \frac{\hbar}{\epsilon(\tilde \mu_b - \tilde \mu_a + V_{b\beta,b\beta} - V_{a\alpha,a\alpha})}$ ($\tau_\S^\epsilon$ is the time characterizing the transition of $\S$ from $\zeta_a$ to another eigenvector induced by the action of $\E$ on $\S$). If $T$ and $\epsilon$ are chosen such that $\theta^\epsilon \sim \tau^\epsilon_\S \ll T$, then $\forall (b\beta) \not=(a\alpha)$, $|\llangle \phi_{a\alpha}|\dot \phi_{b\beta} \rrangle| \ll 1$. In this strong adiabatic regime, as in the very strong adiabatic regime, the system remains projected only onto $\phi_{a\alpha}(t)$ during the whole dynamics. In contrast, if $T$ and $\epsilon$ are chosen such that $\theta^\epsilon \sim \tau^\epsilon_\S \sim T$, then $\forall \beta \not= \alpha$, $|\llangle \phi_{a\alpha}|\dot \phi_{b\beta} \rrangle| \ll 1$ if we assume that $\tau_\E \ll T$, but $|\llangle \phi_{a\alpha}|\dot \phi_{b\alpha} \rrangle| \not\ll 1$. In this weak adiabatic regime, the system remains projected onto the space spanned by the eigenvectors related to $\xi_\alpha$, but transitions between the eigenstate related to $\zeta_a$ to an eigenstate related to another $\zeta_b$ are possible due to non-adiabatic transitions induced by $\E$ on $\S$ (and not directly by the control).\\

In the strong and the very strong adiabatic regimes $\S$ and $\E$ evolve adiabatically with regard to the control (and $\S$ evolves adiabatically with regard to $\E$ in the very strong adiabatic regime). In the weak adiabatic regime, only $\E$ evolves adiabatically with regard to the control, then the evolution of $\S$ can be richer and it is in this case that the adiabatic transport of the density matrix can potentially exhibit operator-valued phases. We note that this weak adiabatic regime is the more interesting from the viewpoint of the quantum control. Indeed, in general, quantum control problems are characterized by the condition $H(T) = H(0)$ since we start and we end with control system off. This induces that $\phi_{a\alpha}(T) = \phi_{a\alpha}(0)$ and in the strong and the very strong adiabatic regimes we have $\rho(T) = \rho(0)$. In contrast due to the possible transitions in the weak adiabatic regime, which are characterized by an operator-valued phase $U \in \mathcal U(\mathcal H_\S)$, we can have $\rho(T) = U \rho(0) U^\dagger$ ($\mathcal U(\mathcal H_\S)$ denotes the set of unitary operators of $\mathcal H_\S$). The answer of the control problem consists then to find the time dependent modulation of the control system such that $U$ transforms $\rho(0)$ to $\rho_{target}$ (or at least such that $\|U \rho(0) U^\dagger - \rho_{target}\|$ to be minimal). The assumption, stating that the evolution of $\E$ must be adiabatic, is natural in this context since it corresponds to require that transitions in $\E$ do not hamper the adiabatic control by generating kinematic decoherence (see section \ref{crossing} and ref. \cite{Viennot3}).\\

The discussion presented here is heuristic, the following section presents rigorous results.

\subsection{Adiabatic theorems}

\begin{theo}[Strong adiabatic theorem for bipartite quantum systems] \label{strongth}
Let \\$[0,1] \ni s \mapsto H(s) = H_\S(s) \otimes 1_\E + 1_\S \otimes H_\E(s) + \epsilon V(s)$ be a family of selfadjoint Hamiltonians of a bipartite quantum system such that $\forall T >0$, $\ihbar \psi'(s) = T H(s) \psi(s)$ has continuous solutions in the norm topology, and such that $V$ is $(H_S \otimes 1_\E + 1_\S \otimes H_\E)$-bounded. Let $\{\mu_b\}_b$ and $\{\nu_\beta\}_\beta$ be the pure point spectra of $H_\S$ and $H_\E$ and $\{\zeta_b\}_b$ and $\{\xi_\beta\}_\beta$ be the associated normalized eigenvectors. Let $\{\phi_{b \beta}\}_{b\beta}$ be the normalized eigenvectors of $H$ continuously linked to $\{\zeta_b \otimes \xi_\beta\}_{b \beta}$ when $\epsilon \to 0$ (in the norm topology). We consider the case where $\psi(0) = \phi_{a\alpha}(0)$. For the sake of simplicity we suppose that each eigenvalue is non degenerate and that $H_\S$ and $H_\E$ do not have continuous spectrum.  We assume the following conditions:
\begin{enumerate}
\item $\forall b,\beta$, $s \mapsto \mu_b(s)$ and $s \mapsto \nu_\beta(s)$ are $\mathcal C^1$;  $s \mapsto \zeta_b(s)$ and $s \mapsto \xi_\beta(s)$ are $\mathcal C^2$ in the norm topology.
\item No resonance between transitions of $\S$ and $\E$ involving $\zeta_a \otimes \xi_\alpha$ occurs, i.e. $\forall s \in [0,1]$, $\forall (b\beta) \not= (a\alpha)$, $\mu_b(s) + \nu_\beta(s) + \epsilon V_{b\beta}(s) \not= \mu_a(s) + \nu_\alpha(s) + \epsilon V_{a\alpha}(s) $.
\item The perturbed energies of $\S$ satisfy a gap condition of order $\epsilon$ with $\mu_a$:
\begin{equation}
\inf_{s \in [0,1]} \min_{(b\beta)\not=(a\alpha)} | \mu_b(s) + \epsilon V_{b\beta,b\beta}(s) - \mu_a(s) - \epsilon V_{a\alpha,a\alpha}(s)| = \mathcal O(\epsilon)
\end{equation} 
\end{enumerate}
Then we have
\begin{equation}
\forall s \in [0,1], \quad P_{a\alpha}(s) \psi(s) = \psi(s) + \mathcal O(\frac{1}{T\epsilon})
\end{equation}
with $P_{a\alpha}(s) = |\phi_{a\alpha}(s)\rrangle \llangle \phi_{a\alpha}(s)|$ the orthogonal projection onto $\phi_{a\alpha}$.
\end{theo}
We remark that we can write also $P_{a\alpha}(s) \psi(s) = \psi(s) + \mathcal O(\frac{\theta^\epsilon}{T})$.\\
\begin{proof}
$\forall s \in [0,1]$, $(\phi_{d\delta}(s))_{d,\delta}$ being a complete basis of the domain of $H(s)$ we can write:
\begin{equation}
\psi(s) = \sum_{d\delta} c_{d\delta}(s) e^{-\ihbar^{-1} T\int_0^s \lambda_{d\delta}(\sigma)d\sigma} \phi_{d\delta}(s)
\end{equation}
for some $c_{d\delta}(s) \in \mathbb C$. By injecting this expression in the Schr\"odinger equation $\ihbar \psi' = TH\psi$ and projecting the result onto $\phi_{b\beta}(s)$, we find
\begin{equation}
c'_{b\beta}(s) = - \sum_{d\delta} c_{d\delta}(s) e^{\ihbar^{-1} T \int_0^s (\lambda_{b\beta}(\sigma)-\lambda_{d\delta}(\sigma))d\sigma} \llangle \phi_{b\beta}(s)|\phi'_{d\delta}(s) \rrangle
\end{equation}
By an integration of this expression with respect to $s$, we find
\begin{eqnarray}
& & c_{b\beta}(s) =  c_{b\beta}(0) \nonumber \\
& & \qquad  - \sum_{d\delta} \int_0^s c_{d\delta}(\sigma) e^{\ihbar^{-1} T \int_0^\sigma (\lambda_{b\beta}(\varsigma)-\lambda_{d\delta}(\varsigma))d\varsigma} \llangle \phi_{b\beta}(\sigma)|\phi'_{d\delta}(\sigma) \rrangle d\sigma
\end{eqnarray}
With an integration by parts we have
\begin{eqnarray}
\label{devc}
& & c_{b\beta}(s)  =  \nonumber \\
& &  c_{b\beta}(0) - \int_0^s c_{b\beta}(\sigma) \llangle \phi_{b\beta}(\sigma)|\phi'_{b\beta}(\sigma)\rrangle d\sigma \nonumber \\
& & - \sum_{(d\delta) \not= (b\beta)} \left( \left[\frac{c_{d\delta}(\sigma) \llangle \phi_{b\beta}(\sigma)|\phi'_{d\delta}(\sigma) \rrangle}{\ihbar^{-1}T (\lambda_{b\beta}(\sigma)-\lambda_{d\delta}(\sigma))} e^{\ihbar^{-1} T \int_0^\sigma (\lambda_{b\beta}(\varsigma)-\lambda_{d\delta}(\varsigma))d\varsigma} \right]^s_0 \right. \nonumber \\
& &  \left. + \int_0^s \frac{e^{\ihbar^{-1} T \int_0^\sigma (\lambda_{b\beta}(\varsigma)-\lambda_{d\delta}(\varsigma))d\varsigma}}{\ihbar^{-1}T}  \left(\frac{c_{d\delta}(\sigma) \llangle \phi_{b\beta}(\sigma)|\phi'_{d\delta}(\sigma) \rrangle}{(\lambda_{b\beta}(\sigma)-\lambda_{d\delta}(\sigma))} \right)' d\sigma \right)
\end{eqnarray}
By a first order perturbation we have $\forall (d\delta) \not= (a\alpha)$
\begin{equation}
\lambda_{a\alpha} - \lambda_{d\delta} = \mu_a - \mu_d + \nu_\alpha - \nu_\delta + \epsilon(V_{a\alpha,a\alpha}-V_{d\delta,d\delta}) + \mathcal O(\epsilon^2)
\end{equation}
Because of the gap and the no resonance conditions, at least $|\lambda_{a\alpha} - \lambda_{d\delta}| = \mathcal O(\epsilon)$ even if $\alpha = \delta$ or $\exists s_*$ such that $\nu_\delta(s_*) = \nu_\alpha(s_*)$. All the other quantities appearing in equation (\ref{devc})  are bounded. Indeed, the eigenvectors being $\mathcal C^2$, $\phi_{d\delta}'$ and $\phi_{d\delta}''$ are defined and bounded on $[0,1]$ ($\sup_{s \in [0,1]} \|\phi_{d\delta}'\|<+\infty$ and $\sup_{s \in [0,1]} \|\phi_{d\delta}''\|<+\infty$) ; the eigenvalues being $\mathcal C^1$, $\lambda_{d\delta}'$ is defined and is bounded on $[0,1]$ ($\sup_{s \in [0,1]} |\lambda_{d\delta}'| < + \infty$) ; moreover $c_{d\delta}<1$ and is $\mathcal C^1$ (because $c_{d\delta} = e^{\ihbar^{-1} T\int_0^s \lambda_{d\delta}d\sigma} \llangle \phi_{d\delta}|\psi\rrangle$ with $\psi$ which is $\mathcal C^1$ as a solution of the Schr\"odinger equation). It follows that
\begin{equation}
\left| \frac{c_{d\delta} \llangle \phi_{b\beta}|\phi'_{d\delta}\rrangle}{\ihbar^{-1}T} \right| \leq \frac{\sup_{s \in [0,1]} \|\phi'_{d\delta}\|}{\hbar^{-1} T} = \mathcal O(\frac{1}{T})
\end{equation}
implying that the third term of equation (\ref{devc}) is $\mathcal O(\frac{1}{T\epsilon})$.
\begin{eqnarray}
\left(\frac{c_{d\delta} \llangle \phi_{b\beta}|\phi'_{d\delta} \rrangle}{\ihbar^{-1}T(\lambda_{b\beta}-\lambda_{d\delta})} \right)' & = & \frac{c_{d\delta}' \llangle \phi_{b\beta} | \phi'_{d\delta} \rrangle}{\ihbar^{-1}T(\lambda_{b\beta}-\lambda_{d\delta})} \nonumber \\
& & + \frac{c_{d\delta} \llangle \phi_{b\beta}' | \phi'_{d\delta} \rrangle}{\ihbar^{-1}T(\lambda_{b\beta}-\lambda_{d\delta})} \nonumber \\
& &+  \frac{c_{d\delta} \llangle \phi_{b\beta}|\phi''_{d\delta} \rrangle}{\ihbar^{-1}T(\lambda_{b\beta}-\lambda_{d\delta})} \nonumber \\
& & - \frac{(\lambda_{b\beta}'-\lambda_{d\delta}') c_{d\delta} \llangle \phi_{b\beta}|\phi'_{d\delta} \rrangle}{\ihbar^{-1}T(\lambda_{b\beta}-\lambda_{d\delta})^2}
\end{eqnarray}
then
\begin{equation}
\left|\frac{c_{d\delta}' \llangle \phi_{b\beta} | \phi'_{d\delta} \rrangle}{\ihbar^{-1}T}\right| \leq \frac{\sup_{s\in[0,1]} |c'_{d\delta}| \sup_{s\in[0,1]} \| \phi'_{d\delta}\|}{\hbar^{-1}T} = \mathcal O(\frac{1}{T})
\end{equation}
\begin{equation}
\left|\frac{c_{d\delta} \llangle \phi_{b\beta}' | \phi'_{d\delta} \rrangle}{\ihbar^{-1}T} \right| \leq \frac{\sup_{s\in[0,1]} \|\phi_{b\beta}'\| \sup_{s\in[0,1]}\| \phi'_{d\delta}\|}{\hbar^{-1}T} = \mathcal O(\frac{1}{T})
\end{equation}
\begin{equation}
\left|\frac{c_{d\delta} \llangle \phi_{b\beta}|\phi''_{d\delta} \rrangle}{\ihbar^{-1}T} \right| \leq \frac{\sup_{s\in[0,1]} \|\phi''_{d\delta}\|}{\hbar^{-1}T}  = \mathcal O(\frac{1}{T})
\end{equation}
\begin{eqnarray}
\left|\frac{(\lambda_{b\beta}'-\lambda_{d\delta}') c_{d\delta} \llangle \phi_{b\beta}|\phi'_{d\delta} \rrangle}{\ihbar^{-1}T} \right| & \leq &  \frac{\sup_{s\in[0,1]} |\lambda_{b\beta}'-\lambda_{d\delta}'| \sup_{s\in[0,1]} \|\phi'_{d\delta}\|}{\hbar^{-1} T} \nonumber \\
& & \qquad = \mathcal O(\frac{1}{T})
\end{eqnarray}
This implies that the fourth term of equation (\ref{devc}) is $\mathcal O(\frac{1}{T\epsilon})$.\\
Finally for $(b\beta)=(a\alpha)$ we have
\begin{equation}
\label{equc}
c_{a\alpha}(s) = 1 - \int_0^s c_{a\alpha}(\sigma) \llangle \phi_{a\alpha}(\sigma)|\phi'_{a\alpha}(\sigma)\rrangle d\sigma + \mathcal O(\frac{1}{T\epsilon})
\end{equation}
and for $\forall (b\beta) \not=(a\alpha)$ we have
\begin{eqnarray}
& & c_{b\beta}(s)  = \nonumber \\
& & - \sum_{(d\delta)\not=(a\alpha)} \int_0^s c_{d\delta}(\sigma) e^{\ihbar^{-1} T \int_0^\sigma (\lambda_{b\beta}(\varsigma)-\lambda_{d\delta}(\varsigma))d\varsigma} \llangle \phi_{b\beta}(\sigma)|\phi'_{d\delta}(\sigma) \rrangle d\sigma \nonumber \\
& & + \mathcal O(\frac{1}{T\epsilon})
\end{eqnarray}
(note that $c_{a\alpha}(0) = 1$ and $c_{b\beta}(0) = 0$). This last expression being the integral equation of Dyson series for null initial condition, we have $c_{b\beta}(s) = \mathcal O(\frac{1}{T\epsilon})$. Finally we conclude that
\begin{equation}
\psi(s) = c_{a\alpha}(s) e^{-\ihbar^{-1} T\int_0^s \lambda_{a\alpha}(\sigma)d\sigma} \phi_{a\alpha}(s) + \mathcal O(\frac{1}{T\epsilon})
\end{equation}
with $c_{a\alpha}(s) = e^{- \int_0^s \llangle \phi_{a\alpha}(\sigma)|\phi'_{a\alpha}(\sigma) \rrangle d\sigma}$ because equation (\ref{equc}) is the integral equation of an exponential map with the initial condition equal to $1$.
\end{proof}
As the very strong adiabatic regime is obtained with an usual adiabatic theorem applied on the bipartite quantum system, it can be considered as a particular case of the previous theorem where the gap condition is stronger and where the remainder of the adiabatic approximation is smaller ($\mathcal O(\frac{1}{T})$ in place of $\mathcal O(\frac{1}{T\epsilon})$).\\
We can remark that assumption \textit{(iii)} is equivalent to $\tau_{\mathcal S} \sim \theta^\epsilon$ as expressed in section \ref{21}. The condition $T \gg \theta^\epsilon$ is equivalent to require that the remainder $\mathcal O(\frac{1}{T\epsilon})$ must be small.\\
Theorem \ref{strongth} can be viewed as a corollary of the usual adiabatic theorem since their assumptions are very similar. Nevertheless, theorem \ref{strongth} corresponds to an adiabatic theorem with a gap condition which is asymptotically small (i.e. $\mathcal O(\epsilon)$). The remainder of the adiabatic approximation ($\mathcal O(\frac{1}{T \epsilon})$) is then larger that the one of the usual adiabatic approximation ($\mathcal O(\frac{1}{T})$). As stated above, the very strong adiabatic regime is a particular case of the theorem \ref{strongth} where the gap is not chosen asymptotically small and corresponds exactly to the usual adiabatic theorem.

\begin{theo}[Weak adiabatic theorem for bipartite quantum systems] \label{weakth}
Let \\ $[0,1] \ni s \mapsto H(s) = H_\S(s) \otimes 1_\E + 1_\S \otimes H_\E(s) + \epsilon V(s)$ be a family of selfadjoint Hamiltonians of a bipartite quantum system such that $\forall T >0$, $\ihbar \psi'(s) = T H(s) \psi(s)$ has continuous solutions in the norm topology, and such that $V$ is $(H_S \otimes 1_\E + 1_\S \otimes H_\E)$-bounded. Let $\{\mu_b\}_b$ and $\{\nu_\beta\}_\beta$ be the pure point spectra of $H_\S$ and $H_\E$ and $\{\zeta_b\}_b$ and $\{\xi_\beta\}_\beta$ be the associated normalized eigenvectors. Let $\{\phi_{b \beta}\}_{b\beta}$ be the normalized eigenvectors of $H$ continuously linked to $\{\zeta_b \otimes \xi_\beta\}_{b \beta}$ when $\epsilon \to 0$ (in the norm topology). We consider the case where $\psi(0) = \phi_{a\alpha}(0)$. For the sake of simplicity we suppose that each eigenvalue is non degenerate and that $H_\S$ and $H_\E$ do not have continuous spectrum.  We assume the following conditions:
\begin{enumerate}
\item $\forall b,\beta$, $s \mapsto \mu_b(s)$ and $s \mapsto \nu_\beta(s)$ are $\mathcal C^1$;  $s \mapsto \zeta_b(s)$ and $s \mapsto \xi_\beta(s)$ are $\mathcal C^2$ in the norm topology.
\item No quasi-resonance between transitions of $\S$ and $\E$ involving $\xi_\alpha$ occurs, i.e. $\forall s \in [0,1]$, $\forall c$, $\forall (b\beta) \not= (c\alpha)$, $|\mu_b(s) + \nu_\beta(s) - \mu_c(s) - \nu_\alpha(s)| = \mathcal O(1)$.
\item The energies of $\mathcal E$ satisfy a gap condition of order 0 with $\nu_\alpha$:
\begin{equation}
\inf_{s\in[0,1]} \min_{\beta \not= \alpha} |\nu_\beta(s) - \nu_\alpha(s)| = \mathcal O(1)
\end{equation}
\item $[0,1] \times \mathbb C \ni (s,z) \mapsto R(s,z)=(H(s)-z)^{-1}$ is strongly $\mathcal C^1$ with respect to $s$ and for every $\delta>0$, $\exists K_\delta \in \mathbb R^{+}$, such that $\|R(s,z)'\| \leq \frac{K_\delta}{\dist(z,\{\lambda_{b\beta}(s)\}_{b\beta})}$ $\forall z$ satisfying $\dist(z,\{\lambda_{b\beta}(s)\}_{b\beta})>\delta$.
\end{enumerate}
Then we have
\begin{equation}
\forall s \in [0,1], \quad P_{\bullet \alpha}(s) \psi(s) = \psi(s) + \mathcal O(\frac{1}{T})
\end{equation}
with $P_{\bullet \alpha}(s) = \sum_b |\phi_{b\alpha}(s)\rrangle \llangle \phi_{b\alpha}(s)|$.
\end{theo}

\begin{proof}
Let $\sigma_\alpha(s) = \{\lambda_{c \alpha}(s)\}_c$ and $\sigma_\bot(s) = \{\lambda_{b\beta}(s)\}_{b,\beta \not= \alpha}$ be a decomposition of the spectrum of $H(s)$ into the part linked to $\nu_\alpha$ and its complementary. By a first order perturbation we have
\begin{equation}
\lambda_{c\alpha} - \lambda_{b\beta} = \mu_c - \mu_b + \nu_\alpha - \nu_\beta + \epsilon(V_{c\alpha,c\alpha}-V_{b\beta,b\beta}) + \mathcal O(\epsilon^2)
\end{equation}
With the gap and the no quasi-resonance conditions we have then
\begin{equation}
\inf_{s \in [0,1]}\dist(\sigma_\alpha(s),\sigma_\bot(s)) = \inf_{s \in [0,1]} \min_{c,b} \min_{\beta \not=\alpha} |\lambda_{c\alpha} - \lambda_{b\beta}| = \mathcal O(1)
\end{equation}
We are in the conditions of the Nenciu adiabatic theorem \cite{Nenciu2} (condition \textit{(iv)} is a requirement of this theorem) which ensures that during the whole evolution, the system remains projected onto the spectral subspace associated with the isolated part of the spectrum $\sigma_\alpha=\{\lambda_{c\alpha}\}_c$. The application of the Nenciu theorem proves the present one which is just a special version.
\end{proof}
We note that the no quasi-resonance condition \textit{(ii)} can be relaxed as a no resonance condition $\mu_b(s) + \nu_\beta(s) \not= \mu_c(s) + \nu_\alpha(s)$ (permitting that $|\mu_b(s) + \nu_\beta(s) - \mu_c(s) - \nu_\alpha(s)| = \mathcal O(\epsilon)$) and moreover it could be suppressed if $V_{c\alpha,c\alpha}(s_*) \not= V_{b\beta,b\beta}(s_*)$ for $s_*$ such that $\mu_b(s_*) + \nu_\beta(s_*) = \mu_c(s_*) + \nu_\alpha(s_*)$. But with these weaker conditions, the remainder of the adiabatic approximation is larger: $\mathcal O(\frac{1}{T\epsilon})$.  Nevertheless the main interest of this theorem is about the weak adiabatic regime $\frac{1}{T\epsilon} \not\in \mathcal V(0)$ (where we cannot apply the strong adiabatic theorem), where the following optional condition is satisfied: \\
\textit{(v) The energies of $\S$ satisfy a gap condition of order $\epsilon$ with $\mu_a$:}
\begin{equation}
\inf_{s\in[0,1]} \min_{b \not=a} |\mu_b(s) - \mu_a(s)| = \mathcal O(\epsilon)
\end{equation}
Condition \textit{(v)} is compatible with the theorem \ref{weakth} but it is not necessary. Nevertheless it corresponds to the interesting physical situations.\\
The assumption \textit{(iii)} implies that $\E$ evolves adiabatically with regard to the control. It is a natural assumption because it corresponds to require that transitions in $\E$ do not hamper the control of $\mathcal S$, as explained at the end of section \ref{21}. But if in practice, it needs to relax this assumption for a single instant (or for a small number of instants) it is possible to generalize the application of the theorem \ref{weakth}, we discuss this point section \ref{crossing}.\\
We can remark that assumption \textit{(iii)} is equivalent to $T \gg \tau_\E$ as expressed in section \ref{21}, whereas condition \textit{(v)} is equivalent to $\tau_\S \sim \theta^\epsilon$.

\section{Adiabatic transport of the density matrix}
We are now able to find adiabatic transport formulae for the density matrix of $\S$ : $\rho(s) = \tr_\E |\psi(s) \rrangle \llangle \psi(s)|$. 

\subsection{Strong adiabatic regime}
\begin{propo}
In the conditions of the strong adiabatic theorem (theorem \ref{strongth}) we have
\begin{equation}
\forall s\in [0,1] \quad \rho(s) = \rho_{a\alpha}(s) + \mathcal O(\frac{1}{T\epsilon})
\end{equation}
where $\rho_{a\alpha}(s) = \tr_\E |\phi_{a\alpha}(s) \rrangle \llangle \phi_{a\alpha}(s)|$ is the ``density eigenmatrix''.
\end{propo}

\begin{proof}
By applying theorem \ref{strongth} we have
\begin{equation}
\psi(s) = e^{-\ihbar^{-1} T \int_0^s \lambda_{a\alpha}(\sigma)d\sigma - \int_0^s \llangle \phi_{a\alpha}(\sigma)|\phi'_{a\alpha}(\sigma) \rrangle d\sigma} \phi_{a\alpha}(s) + \mathcal O(\frac{1}{T\epsilon})
\end{equation}
Since $\llangle \phi_{a\alpha}(s)|\phi_{a\alpha}(s) \rrangle = 1 \Rightarrow \llangle \phi_{a\alpha}(s)|\phi'_{a\alpha}(s) \rrangle \in \imath \mathbb R$, we have $|\psi(s) \rrangle \llangle \psi(s)| = |\phi_{a\alpha}(s) \rrangle \llangle \phi_{a\alpha}(s)|$.
\end{proof}
We can approach $\rho_{a\alpha}$ by a perturbative method (using the Wigner-Brillouin approach):
\begin{eqnarray}
\rho(s) & = & |\zeta_a(s)\rangle\langle \zeta_a(s)| + \epsilon \sum_{b\not=a} \frac{V_{b\alpha,a\alpha}(s)}{\mu_a(s)-\mu_b(s)+\epsilon V_{a\alpha,a\alpha}} |\zeta_b(s)\rangle\langle \zeta_a(s)| \nonumber \\
& &  \quad + \epsilon \sum_{b\not=a} \frac{V_{a\alpha,b\alpha}(s)}{\mu_b(s)-\mu_a(s)+\epsilon V_{a\alpha,a\alpha}} |\zeta_a(s)\rangle\langle \zeta_b(s)| \nonumber \\
& & \quad + \mathcal O(\max(\frac{1}{T\epsilon},\epsilon^2))
\end{eqnarray}

\subsection{Zero order weak adiabatic regime}
We denote by $\Teg$ and $\Ted$ the time ordered and the time anti-ordered exponentials, i.e. for $s \mapsto A(s)$ a bounded anti-selfadjoint operator, $\Teg^{- \int_0^s A(\sigma)d\sigma}$ is the unitary operator solution of
\begin{equation}
\left(\Teg^{- \int_0^s A(\sigma)d\sigma} \right)' = - A(s) \Teg^{- \int_0^s A(\sigma)d\sigma} \qquad \Teg^{- \int_0^0 A(\sigma)d\sigma} = 1
\end{equation}
and $\Ted^{- \int_0^s A(\sigma)d\sigma}$ is the unitary operator solution of
\begin{equation}
\left(\Ted^{- \int_0^s A(\sigma)d\sigma} \right)' = - \Ted^{- \int_0^s A(\sigma)d\sigma} A(s) \qquad \Ted^{- \int_0^0 A(\sigma)d\sigma} = 1
\end{equation}
Moreover we denote by $\Ad$ the adjoint action of a transformation $U$ on a density matrix $\rho$:
\begin{equation}
\Ad[U] \rho = U \rho U^\dagger
\end{equation}

\begin{propo}
In the conditions of the weak adiabatic theorem (theorem \ref{weakth}) we have $\forall s \in [0,1]$
\begin{eqnarray}
\rho(s) & = & \Ad \left[\Teg^{-\ihbar^{-1} T \int_0^s E_\alpha^{(0)}(\sigma)d\sigma} \Ted^{- \int_0^s A^{(0)} (\sigma)d\sigma} \right] \rho_{a\alpha}(s) \nonumber \\
& & \qquad + \mathcal O(\max(\frac{1}{T},\epsilon))
\end{eqnarray}
with the zero order dynamical phase generator defined as being
\begin{equation}
E_\alpha^{(0)}(s) = \sum_b \lambda_{b\alpha}(s) |\zeta_b(s) \rangle \langle \zeta_b(s)| \in \mathcal L(\mathcal H_\S)
\end{equation}
and the zero order geometric phase generator defined as being
\begin{equation}
A^{(0)}(s) = \sum_{b,c} \langle \zeta_b(s)|\zeta_c'(s) \rangle |\zeta_b(s) \rangle \langle \zeta_c(s)| \in \mathcal L(\mathcal H_\S)
\end{equation}
\end{propo}

\begin{proof}
By applying the theorem \ref{weakth} and an adiabatic transport formula for several eigenvalues \cite{Viennot4,Viennot5} we have
\begin{equation}
\psi(s) = \sum_b \left[\Teg^{-\ihbar^{-1} T \int_0^s \Lambda_\alpha d\sigma - \int_0^s K_\alpha d\sigma} \right]_{ba} \phi_{b \alpha} + \mathcal O(\frac{1}{T})
\end{equation}
with $\Lambda_\alpha, K_\alpha \in \mathfrak M_{n \times n}(\mathbb C)$ being the square matrices of order $n$ ($n$ is the dimension of $\mathcal H_\S$) such that $[\Lambda_\alpha]_{ab} = \lambda_{a\alpha} \delta_{ab}$ and $[K_\alpha]_{ab} = \llangle \phi_{a\alpha}|\phi_{b\alpha}' \rrangle$ ($[.]_{ab}$ denotes the matrix element at the $a$-th line and the $b$-th column). Since $|\zeta_b \rangle \langle \zeta_c| \phi_{a\alpha} = \phi_{b\alpha} \delta_{ca} + \mathcal O(\epsilon)$ we have
\begin{equation}
\psi(s) = \sum_{b,c} \left[\Teg^{-\ihbar^{-1} T \int_0^s \Lambda_\alpha d\sigma - \int_0^s K_\alpha d\sigma} \right]_{bc}|\zeta_b \rangle \langle \zeta_c| \phi_{a \alpha} + \mathcal O(\max(\frac{1}{T},\epsilon))
\end{equation}
By applying corollary \ref{cor} (\ref{corap}) we have
\begin{eqnarray}
& & \sum_{b,c} \left[\Teg^{-\ihbar^{-1} T \int_0^s \Lambda_\alpha d\sigma - \int_0^s K_\alpha d\sigma} \right]_{bc}|\zeta_b \rangle \langle \zeta_c| \nonumber\\
& & \qquad  = \sum_{bd} \left[\Teg^{-\int_0^s Xd\sigma} \right]_{bd} |\zeta_b \rangle \langle \zeta_d| \sum_{f,c} \left[\Teg^{-\int_0^s K_\alpha d\sigma} \right]_{fc} |\zeta_f \rangle \langle \zeta_c|
\end{eqnarray}
with $X = \ihbar^{-1}T \Lambda_\alpha + K_\alpha - \Teg^{-\int_0^s Xd\sigma} K_\alpha \left(\Teg^{-\int_0^s Xd\sigma}\right)^{-1}$.
\begin{itemize}
\item Let $Y \in \mathfrak M_{n \times n}(\mathbb C)$ be such that $\Teg^{-\ihbar^{-1} T \int_0^s E_\alpha^{(0)} d\sigma} = \sum_{b,d} \left[\Teg^{-\int_0^s Y d\sigma} \right]_{bd} |\zeta_b \rangle \langle \zeta_d|$.
\begin{equation}
\left(\Teg^{-\ihbar^{-1} T\int_0^s E_\alpha^{(0)} d\sigma} \right)' = - \ihbar^{-1}T E_\alpha^{(0)} \Teg^{-\ihbar^{-1} T\int_0^s E_\alpha^{(0)} d\sigma}
\end{equation}
implies that
\begin{eqnarray}
& & - \ihbar^{-1}T E_\alpha^{(0)}  \sum_{b,d} \left[\Teg^{-\int_0^s Y d\sigma} \right]_{bd} |\zeta_b \rangle \langle \zeta_d| \nonumber \\
& & \qquad = \sum_{b,d} \left[-Y\Teg^{-\int_0^s Y d\sigma} \right]_{bd} |\zeta_b \rangle \langle \zeta_d| \nonumber \\
& & \qquad \qquad +  \sum_{b,d} \left[\Teg^{-\int_0^s Y d\sigma} \right]_{bd} |\zeta_b' \rangle \langle \zeta_d| \nonumber \\
& & \qquad \qquad + \sum_{b,d} \left[\Teg^{-\int_0^s Y d\sigma} \right]_{bd} |\zeta_b \rangle \langle \zeta_d'| \label{E0}
\end{eqnarray}
but
\begin{eqnarray}
& & - \ihbar^{-1}T E_\alpha^{(0)}  \sum_{b,d} \left[\Teg^{-\int_0^s Y d\sigma} \right]_{bd} |\zeta_b \rangle \langle \zeta_d| \nonumber \\
&  & \qquad = - \ihbar^{-1}T  \sum_{b,d}\lambda_{b\alpha}  \left[\Teg^{-\int_0^s Y d\sigma} \right]_{bd} |\zeta_b \rangle \langle \zeta_d| \\
& & \qquad = \sum_{b,d} \left[- \ihbar^{-1} T \Lambda_\alpha \Teg^{-\int_0^s Y d\sigma} \right]_{bd} |\zeta_b \rangle \langle \zeta_d|
\end{eqnarray}
and
\begin{eqnarray}
|\zeta_b'\rangle \langle \zeta_d| & = & \sum_f \langle \zeta_f|\zeta_b' \rangle |\zeta_f\rangle \langle \zeta_d| \\
|\zeta_b\rangle \langle \zeta_d'| & = & \sum_f \langle \zeta_d'|\zeta_f \rangle |\zeta_b\rangle \langle \zeta_f|
\end{eqnarray}
Equation (\ref{E0}) becomes
\begin{eqnarray}
-\ihbar^{-1} T \Lambda_\alpha \Teg^{-\int_0^s Yd\sigma} & = & -Y \Teg^{-\int_0^s Yd\sigma} + \mathring K \Teg^{-\int_0^s Yd\sigma} \nonumber \\
& & \qquad - \Teg^{-\int_0^s Yd\sigma} \mathring K \label{Y0}
\end{eqnarray}
because $\langle \zeta_d|\zeta_f \rangle = \delta_{df} \Rightarrow \langle \zeta_d'|\zeta_f \rangle = - \langle \zeta_d|\zeta_f'\rangle$ with $\mathring K \in \mathfrak M_{n\times n}(\mathbb C)$ defined as $[\mathring K]_{df} = \langle \zeta_d|\zeta_f'\rangle$. But
\begin{eqnarray}
[K_\alpha]_{bc} & = & \llangle \phi_{b\alpha}|\phi_{c\alpha}' \rrangle \\
& = & \langle \zeta_b|\zeta_c'\rangle + \langle \xi_\alpha|\xi_\alpha'\rangle + \mathcal O(\epsilon)
\end{eqnarray}
\begin{equation}
\Rightarrow K_\alpha = \mathring K + \langle \xi_\alpha|\xi_\alpha'\rangle + \mathcal O(\epsilon)
\end{equation}
and then $X = \ihbar^{-1} T \Lambda_\alpha + \mathring K - \Teg^{-\int_0^s Xd\sigma} \mathring K \left(\Teg^{-\int_0^s Xd\sigma}\right)^{-1} + \mathcal O(\epsilon)$ ($\langle \xi_\alpha|\xi_\alpha'\rangle \in \imath \mathbb R$). A comparison with equation (\ref{Y0}) shows that $Y = X +\mathcal O(\epsilon)$ and then
\begin{equation}
 \sum_{bd} \left[\Teg^{-\int_0^s Xd\sigma} \right]_{bd} |\zeta_b \rangle \langle \zeta_d| = \Teg^{-\ihbar^{-1}T \int_0^s E_\alpha^{(0)}d\sigma} + \mathcal O(\epsilon)
\end{equation}
\item Let $Z \in \mathfrak M_{n\times n}(\mathbb C)$ be such that $\Ted^{-\int_0^s A^{(0)} d\sigma} = \sum_{f,c} \left[\Teg^{-\int_0^s Zd\sigma}\right]_{f,c} |\zeta_f \rangle \langle \zeta_c|$.
\begin{equation}
\left(\Ted^{-\int_0^s A^{(0)} d\sigma} \right)' = - \Ted^{-\int_0^s A^{(0)} d\sigma} A^{(0)}
\end{equation}
implies that
\begin{eqnarray}
& & - \sum_{f,c} \left[\Teg^{-\int_0^s Z d\sigma} \right]_{fc} |\zeta_f \rangle \langle \zeta_c| A^{(0)} \nonumber \\
& & \qquad = \sum_{f,c} \left[-Z\Teg^{-\int_0^s Z d\sigma} \right]_{fc} |\zeta_f \rangle \langle \zeta_c| \nonumber \\
& & \qquad \qquad +  \sum_{f,c} \left[\Teg^{-\int_0^s Z d\sigma} \right]_{fc} |\zeta_f' \rangle \langle \zeta_c| \nonumber \\
& & \qquad \qquad + \sum_{f,c} \left[\Teg^{-\int_0^s Z d\sigma} \right]_{fc} |\zeta_f \rangle \langle \zeta_c'|
\end{eqnarray}
which becomes (since $A^{(0)} = \sum_{cg} [\mathring K]_{cg} |\zeta_c \rangle \langle \zeta_g|$):
\begin{equation}
- \Teg^{-\int_0^s Z d\sigma} \mathring K  = - Z \Teg^{-\int_0^s Z d\sigma} + \mathring K \Teg^{-\int_0^s Z d\sigma} - \Teg^{-\int_0^s Z d\sigma} \mathring K
\end{equation}
$Z = \mathring K = K_\alpha -\langle \xi_\alpha|\xi_\alpha'\rangle + \mathcal O(\epsilon)$ and then
\begin{equation}
\sum_{f,c} \left[\Teg^{-\int_0^s K_\alpha d\sigma} \right]_{fc} |\zeta_f \rangle \langle \zeta_c| = e^{-\int_0^s \langle \xi_\alpha|\xi_\alpha'\rangle d\sigma} \Ted^{-\int_0^s A^{(0)}d\sigma} + \mathcal O(\epsilon)
\end{equation} 
\end{itemize}
Finally we have
\begin{eqnarray}
\psi(s) & = & e^{-\int_0^s \langle \xi_\alpha|\xi_\alpha'\rangle d\sigma} \Teg^{-\ihbar^{-1} T\int_0^s E_\alpha^{(0)} d\sigma} \Ted^{-\int_0^s A^{(0)}d\sigma} \phi_{a\alpha}(s) \nonumber \\
& & \qquad + \mathcal O(\max(\frac{1}{T},\epsilon))
\end{eqnarray}
and
\begin{eqnarray}
\rho(s) & = & \tr_\E |\psi(s) \rrangle \llangle \psi(s)| \\
& = & \Ad\left[\Teg^{-\ihbar^{-1} T\int_0^s E_\alpha^{(0)} d\sigma} \Ted^{-\int_0^s A^{(0)}d\sigma} \right] \rho_{a\alpha} + \mathcal O(\max(\frac{1}{T},\epsilon))
\end{eqnarray}
since $e^{-\int_0^s \langle \xi_\alpha|\xi_\alpha'\rangle d\sigma} \in U(1)$ ($U(1)$ is the set of unit modulus complex numbers) and $\Teg^{-\ihbar^{-1} T\int_0^s E_\alpha^{(0)} d\sigma}, \Ted^{-\int_0^s A^{(0)}d\sigma} \in \mathcal U(\mathcal H_\S)$.
\end{proof}
In this zero order approximation, the only memory of $\E$ are the elements $T\lambda_{b\alpha}$ in the expression of the operator-valued dynamical phase. We note that we cannot approach $\lambda_{b\alpha}$ at the zero order perturbative approximation in the dynamical phase because $T\epsilon$ is not negligible in the weak adiabatic regime. In the next section we consider higher accuracy approximations.

\subsection{First order weak adiabatic regime}
\begin{propo}
In the conditions of the weak adiabatic theorem (theorem \ref{weakth}) we have $\forall s \in[0,1]$
\begin{eqnarray}
\rho(s) & = & \Ad\left[\Teg^{-\ihbar^{-1} T \int_0^s E_\alpha^{(1)}(\sigma)d\sigma} \Ted^{- \int_0^s A^{(1)}_\alpha (\sigma)d\sigma} \right] \rho_{a\alpha}(s) \nonumber \\
& & \qquad + \mathcal O(\max(\frac{1}{T},\epsilon^2))
\end{eqnarray}
with the first order dynamical phase generator defined as being
\begin{equation}
E_\alpha^{(1)}(s) = \sum_{b,c} \left(\lambda_{b\alpha}(s) \delta_{bc} - \frac{\ihbar}{T} \eta_{\alpha bc}^{(1)}(s)\right) |\zeta_{b\alpha}^{(1)}(s) \rangle \langle \zeta_{c\alpha}^{(1)}(s)| \in \mathcal L(\mathcal H_\S)
\end{equation}
and the first order geometric phase generator defined as being
\begin{equation}
A^{(1)}_\alpha(s) = \sum_{b,c} \langle \zeta_{b\alpha}^{(1)}(s)|\zeta_{c\alpha}^{(1)\prime}(s) \rangle |\zeta_{b\alpha}^{(1)}(s) \rangle \langle \zeta_{c\alpha}^{(1)}(s)| \in \mathcal L(\mathcal H_\S)
\end{equation}
with
\begin{equation}
\zeta_{b\alpha}^{(1)}(s) = \zeta_b(s) + \epsilon \sum_{d\not=b} \frac{V_{d\alpha,b\alpha}(s)}{\mu_b(s)-\mu_d(s)+\epsilon V_{b\alpha,b\alpha}(s)} \zeta_d(s)
\end{equation}
\begin{eqnarray}
\eta^{(1)}_{\alpha bc}(s) & = & \langle \xi_\alpha(s)|\xi_\alpha'(s)\rangle \delta_{bc} \nonumber \\
& & \quad + \epsilon \sum_{\gamma \not=\alpha} \frac{V_{b\gamma,c\alpha}(s) \langle\xi_\alpha(s)|\xi_\gamma'(s)\rangle (1-\delta_{bc})}{\mu_c(s)-\mu_b(s)+\nu_\alpha(s)-\nu_\gamma(s) + \epsilon V_{c\alpha,c\alpha}(s)} \nonumber \\
& & \quad + \epsilon \sum_{\gamma \not=\alpha} \frac{V_{b\alpha,c\gamma}(s) \langle\xi_\gamma(s)|\xi_\alpha'(s)\rangle (1-\delta_{bc})}{\mu_b(s)-\mu_c(s)+\nu_\alpha(s)-\nu_\gamma(s) + \epsilon V_{b\alpha,b\alpha}(s)}
\end{eqnarray}
\end{propo}

\begin{proof}
As for the zero order formula we start with
\begin{eqnarray}
\psi(s) & = & \sum_b \left[\Teg^{-\ihbar^{-1} T \int_0^s \Lambda_\alpha d\sigma - \int_0^s K_\alpha d\sigma} \right]_{ba} \phi_{b \alpha} + \mathcal O(\frac{1}{T}) \\
& = & \sum_{b,c} U_{\alpha bc} |\phi_{b\alpha} \rrangle \llangle \phi_{c\alpha}| \phi_{a\alpha} + \mathcal O(\frac{1}{T}) \label{psi1}
\end{eqnarray}
Let $\zeta_{b\alpha}^{(1)} = \zeta_b + \epsilon \sum_{d\not=b} \mathcal V_{d\alpha,b\alpha} \zeta_d$ (with $\mathcal V_{d\gamma,b\alpha} = \frac{V_{d\gamma,b\alpha}}{\mu_b-\mu_d + \nu_\alpha-\nu_\gamma+\epsilon V_{b\alpha,b\alpha}}$). By using a first order perturbative expansion based on the Wigner-Brillouin method we have
\begin{eqnarray}
\phi_{b\alpha} & = & \zeta_b \otimes \xi_\alpha+\epsilon \sum_{(d,\gamma)\not=(b,\alpha)} \mathcal V_{d\gamma,b\alpha} \zeta_d \otimes \xi_\gamma + \mathcal O(\epsilon^2) \\
& = & \zeta_{b\alpha}^{(1)} \otimes \xi_\alpha + \epsilon \sum_{\gamma \not=\alpha} \sum_{d\not=b} \mathcal V_{d\gamma,b\alpha} \zeta_d \otimes \xi_\gamma + \mathcal O(\epsilon^2)
\end{eqnarray}
We have then
\begin{eqnarray}
|\phi_{b\alpha}\rrangle \llangle \phi_{c\alpha}| & = & |\zeta_{b\alpha}^{(1)} \rangle \langle \zeta_{c\alpha}^{(1)}| \otimes |\xi_\alpha\rangle \langle \xi_\alpha| \nonumber \\
& &  \quad + \epsilon \sum_{\gamma\not=\alpha} \sum_{d\not=b} \mathcal V_{d\gamma,b\alpha} |\zeta_d \rangle \langle \zeta_c| \otimes |\xi_\gamma\rangle \langle \xi_\alpha| \nonumber \\
& & \quad + \epsilon \sum_{\gamma\not=\alpha} \sum_{d\not=c} \overline{\mathcal V_{d\gamma,c\alpha}} |\zeta_b \rangle \langle \zeta_d| \otimes |\xi_\alpha \rangle \langle \xi_\gamma| + \mathcal O(\epsilon^2)
\end{eqnarray}
\begin{eqnarray}
|\phi_{b\alpha}\rrangle \llangle \phi_{c\alpha}|\phi_{a\alpha} & = & |\zeta_{b\alpha}^{(1)} \rangle \langle \zeta_{c\alpha}^{(1)}| \phi_{a\alpha} \nonumber \\
& &  \quad - \epsilon \sum_{\gamma\not=\alpha} \mathcal V_{c\gamma,a\alpha} \zeta_b \otimes \xi_\gamma (1-\delta_{ac}) \nonumber \\
& & \quad + \epsilon \sum_{\gamma\not=\alpha} \sum_{d\not=b} \mathcal V_{d\gamma,b\alpha} \zeta_d \otimes \xi_\gamma \delta_{ac} + \mathcal O(\epsilon^2) \\
& = & |\zeta_{b\alpha}^{(1)} \rangle \langle \zeta_{c\alpha}^{(1)}| \phi_{a\alpha} \nonumber \\
& &  \quad - \epsilon \sum_{\gamma\not=\alpha} \sum_{d\not=c} \mathcal V_{c\gamma,d\alpha} |\zeta_b\rangle\langle \zeta_d|\phi_{a\gamma} \nonumber \\
& & \quad + \epsilon \sum_{\gamma\not=\alpha} \sum_{d\not=b} \mathcal V_{d\gamma,b\alpha} |\zeta_d\rangle \langle \zeta_c|\phi_{a\gamma} + \mathcal O(\epsilon^2)
\end{eqnarray}
By using this expression with the equation (\ref{psi1}) we find that
\begin{eqnarray}
\psi(s) & = & \sum_{bc} U_{\alpha bc} |\zeta_{b\alpha}^{(1)}\rangle \langle \zeta_{b\alpha}^{(1)}| \phi_{a\alpha} \nonumber \\
& & \quad - \epsilon \sum_{b,c} \sum_{d\not=c} \sum_{\gamma\not=\alpha} U_{\alpha bc} \mathcal V_{c\gamma,d\alpha} |\zeta_b \rangle \langle \zeta_d| \phi_{a\gamma} \nonumber \\
& & \quad + \epsilon \sum_{b,c} \sum_{d\not=b} \sum_{\gamma \not=\alpha} U_{\alpha bc} \mathcal V_{d\gamma,b\alpha} |\zeta_d\rangle \langle \zeta_c| \phi_{a\gamma} + \mathcal O(\max(\frac{1}{T},\epsilon^2)) \\
& = & \mathcal U_\alpha \phi_{a\alpha} + \epsilon \sum_{\gamma\not=\alpha} [\mathcal W_{\gamma\alpha},\mathcal U_\alpha] \phi_{a\gamma} + \mathcal O(\max(\frac{1}{T},\epsilon^2))
\end{eqnarray}
with the operators of $\mathcal H_\S$ : $\mathcal U_\alpha = \sum_{b,c} U_{\alpha bc} |\zeta_{b\alpha}^{(1)}\rangle \langle \zeta_{c\alpha}^{(1)}|$ and $\mathcal W_{\gamma \alpha} = \sum_d \sum_{b\not=d} \mathcal V_{d\gamma,b\alpha}|\zeta_{d\gamma}^{(1)}\rangle \langle \zeta_{b\alpha}^{(1)}|$. We have then
\begin{eqnarray}
|\psi \rrangle \llangle \psi| & = & \mathcal U_\alpha |\phi_{a\alpha} \rrangle \llangle \phi_{a\alpha}| \mathcal U_\alpha^\dagger \nonumber \\
& & \quad + \epsilon \sum_{\gamma\not=\alpha} [\mathcal W_{\gamma \alpha},\mathcal U_\alpha] |\zeta_a \rangle \langle \zeta_a| \otimes |\xi_\gamma\rangle \langle \xi_\alpha| \mathcal U_\alpha^\dagger \nonumber \\
& & \quad + \epsilon \sum_{\gamma\not=\alpha} \mathcal U_\alpha |\zeta_a \rangle \langle \zeta_a| \otimes |\xi_\alpha \rangle \langle \xi_\gamma| [\mathcal U_\alpha^\dagger,\mathcal W_{\gamma \alpha}^\dagger] \nonumber \\
& & \qquad + \mathcal O(\max(\frac{1}{T},\epsilon^2))
\end{eqnarray}
\begin{equation}
\Rightarrow \rho(s) = \tr_\E |\psi\rrangle \llangle \psi| = \mathcal U_\alpha \rho_{a\alpha} \mathcal U_\alpha^\dagger + \mathcal O(\max(\frac{1}{T},\epsilon^2)) \label{rho1}
\end{equation}
\begin{equation}
\mathcal U_\alpha = \sum_{b,c} \left[\Teg^{-\ihbar^{-1} T \int_0^s \Lambda_\alpha d\sigma - \int_0^s K_\alpha d\sigma} \right]_{bc}|\zeta^{(1)}_{b\alpha}\rangle\langle \zeta^{(1)}_{c\alpha}|
\end{equation}
But
\begin{eqnarray}
[K_\alpha]_{bc} & = & \llangle \phi_{b\alpha}|\phi_{c\alpha}'\rrangle \\
& = & \langle \zeta_{b\alpha}^{(1)}|\zeta_{c\alpha}^{(1)\prime}\rangle + \langle \xi_\alpha|\xi_\alpha'\rangle \nonumber \\
& & \quad + \epsilon \sum_{\gamma\not=\alpha} \mathcal V_{b\gamma,c\alpha} \langle \xi_\alpha|\xi_\gamma'\rangle (1-\delta_{bc}) \nonumber \\
& & \quad + \epsilon \sum_{\gamma\not=\alpha} \overline{\mathcal V_{c\gamma,b\alpha}} \langle \xi_\gamma|\xi_\alpha'\rangle (1-\delta_{bc}) +\mathcal O(\epsilon^2) \\
& = & [\mathring K_\alpha]_{bc} + \eta^{(1)}_{\alpha bc} + \mathcal O(\epsilon^2)
\end{eqnarray}
with $[\mathring K_\alpha]_{bc} = \langle \zeta_{b\alpha}^{(1)}|\zeta_{c\alpha}^{(1)\prime}\rangle$ ($\mathring K_\alpha \in \mathfrak M_{n\times n}(\mathbb C)$). By using the corollary \ref{cor} (\ref{corap}) we find
\begin{eqnarray}
& & \mathcal U_\alpha \nonumber \\
& & = \sum_{b,d} \left[\Teg^{-\int_0^s Xd\sigma}\right]_{bd} |\zeta^{(1)}_{b\alpha}\rangle \langle \zeta^{(1)}_{d\alpha}| \sum_{f,c} \left[\Teg^{-\int_0^s \mathring K_\alpha d\sigma} \right]_{fc} |\zeta^{(1)}_{f\alpha}\rangle \langle \zeta^{(1)}_{c\alpha}| \nonumber \\
& & + \mathcal O(\epsilon^2)
\end{eqnarray}
with $X = \ihbar^{-1}T \Lambda_\alpha + \eta^{(1)}_\alpha + \mathring K_\alpha - \Teg^{-\int_0^s X d\sigma} \mathring K_\alpha \left(\Teg^{-\int_0^s X d\sigma} \right)^{-1}$.
\begin{itemize}
\item Let $Y \in \mathfrak M_{n\times n}(\mathbb C)$ be such that $\Teg^{-\ihbar^{-1}T\int_0^s E_\alpha^{(1)} d\sigma} = \sum_{b,d} \left[\Teg^{-\int_0^s Yd\sigma}\right]_{bd} |\zeta^{(1)}_{b\alpha}\rangle \langle \zeta^{(1)}_{d\alpha}|$
\begin{equation}
\left(\Teg^{-\ihbar^{-1}T\int_0^s E_\alpha^{(1)} d\sigma}\right)' = -\ihbar^{-1} T E_\alpha^{(1)} \Teg^{-\ihbar^{-1}T\int_0^s E_\alpha^{(1)} d\sigma}
\end{equation}
implies that
\begin{eqnarray}
& & \sum_{b,d} \left[(-\ihbar^{-1}T\Lambda_\alpha-\eta_\alpha^{(1)})\Teg^{-\int_0^s Yd\sigma} \right]_{bd} |\zeta^{(1)}_{b\alpha}\rangle \langle \zeta^{(1)}_{d\alpha}| \nonumber \\
& & \quad = \sum_{b,d} \left[-Y\Teg^{-\int_0^s Yd\sigma}\right]_{bd} |\zeta^{(1)}_{b\alpha}\rangle \langle \zeta^{(1)}_{d\alpha}| \nonumber \\
& & \qquad + \sum_{b,d,f} \left[\Teg^{-\int_0^s Yd\sigma}\right]_{bd} \langle \zeta_{f\alpha}^{(1)}|\zeta_{b\alpha}^{(1)\prime}\rangle |\zeta^{(1)}_{f\alpha}\rangle \langle \zeta^{(1)}_{d\alpha}| \nonumber \\
& & \qquad + \sum_{b,d,f} \left[\Teg^{-\int_0^s Yd\sigma}\right]_{bd} \langle \zeta_{d\alpha}^{(1) \prime}|\zeta_{f\alpha}^{(1)}\rangle |\zeta^{(1)}_{b\alpha}\rangle \langle \zeta^{(1)}_{f\alpha}|
\end{eqnarray}
\begin{eqnarray}
& \Rightarrow & (-\ihbar^{-1}T\Lambda_\alpha-\eta_\alpha^{(1)}) \Teg^{-\int_0^s Yd\sigma} \nonumber \\
&& \quad = -Y\Teg^{-\int_0^s Yd\sigma}+\mathring K_\alpha \Teg^{-\int_0^s Yd\sigma} - \Teg^{-\int_0^s Yd\sigma} \mathring K_\alpha
\end{eqnarray}
We have then $Y = \ihbar^{-1}T\Lambda_{\alpha}+\eta_\alpha^{(1)} + \mathring K_\alpha - \Teg^{-\int_0^s Yd\sigma}\mathring K_\alpha \left(\Teg^{-\int_0^s Yd\sigma} \right)^{-1}$ and by comparison with the definition of $X$ we have $Y=X$ and then
\begin{equation}
\sum_{b,d} \left[\Teg^{-\int_0^s Xd\sigma}\right]_{bd} |\zeta^{(1)}_{b\alpha}\rangle \langle \zeta^{(1)}_{d\alpha}| = \Teg^{-\ihbar^{-1}T\int_0^s E_\alpha^{(1)} d\sigma}
\end{equation}
\item Let $Z \in \mathfrak M_{n\times n}(\mathbb C)$ be such that $\Ted^{-\int_0^s A_\alpha^{(1)}d\sigma} = \sum_{f,c} \left[\Teg^{-\int_0^s Zd\sigma}\right]_{fc} |\zeta_{f\alpha}^{(1)}\rangle \langle \zeta_{c\alpha}^{(1)}|$.
\begin{equation}
\left(\Ted^{-\int_0^s A_\alpha^{(1)}d\sigma}\right)' = - \Ted^{-\int_0^s A_\alpha^{(1)}d\sigma} A_\alpha^{(1)}
\end{equation}
 implies that
\begin{eqnarray}
- \Teg^{-\int_0^s Zd\sigma} \mathring K_\alpha & = & - Z \Teg^{-\int_0^s Zd\sigma} \nonumber \\
& & \quad + \mathring K_\alpha \Teg^{-\int_0^s Zd\sigma} - \Teg^{-\int_0^s Zd\sigma}  \mathring K_\alpha
\end{eqnarray}
$Z = \mathring K_\alpha$ and then
\begin{equation}
\sum_{f,c} \left[\Teg^{-\int_0^s \mathring K_\alpha d\sigma}\right]_{fc} |\zeta_{f\alpha}^{(1)}\rangle \langle \zeta_{c\alpha}^{(1)}| = \Ted^{-\int_0^s A_\alpha^{(1)}d\sigma}
\end{equation}
\end{itemize}
Finally we have
\begin{equation}
\mathcal U_\alpha = \Teg^{-\ihbar^{-1}T\int_0^s E_\alpha^{(1)} d\sigma}  \Ted^{-\int_0^s A_\alpha^{(1)}d\sigma} + \mathcal O(\epsilon^2)
\end{equation}
This concludes the proof by injecting this expression in equation (\ref{rho1}).
\end{proof}
We note that we have used the Wigner-Brillouin method for the perturbation theory ($\zeta_{b\alpha}^{(1)}(s) = \zeta_b(s) + \epsilon \sum_{d\not=b} \frac{V_{d\alpha,b\alpha}(s)}{\mu_b(s)-\mu_d(s)+\epsilon V_{b\alpha,b\alpha}(s)} \zeta_d(s)$) because the weak adiabatic regime does not need that $\tau_\S \ll T$, and permits crossings of eigenvalues of $\S$. The Wigner-Brillouin method permits to avoid some divergences in the perturbation expansion induced by these possible crossings.\\
We can remark that the off-diagonal part of $E_\alpha^{(1)}$ includes in fact a geometric phase generator associated with $\E$ ($\eta_\alpha^{(1)}$). This is not surprising since $E_\alpha^{(1)}$ takes the role of an effective Hamiltonian of $\S$ dressed by $\E$ (see below). Such a dynamical phase generator is similar with the one generated by quasi-energies in adiabatic Floquet theory (see for example \cite{Viennot6}) where the role of $\S$ is played by an atom or a molecule interacting with a strong laser field described by $L^2(S^1,\frac{d\theta}{2\pi})$ (the space of square integrable functions on the circle $S^1$, $\theta$ being the laser phase) which plays the role of $\mathcal H_\E$.

\subsection{Second order weak adiabatic regime}
The case of the second order perturbative approximation is more difficult. Indeed the second order approximation of the eigenvectors:
\begin{eqnarray}
\phi_{b\beta} & = & \zeta_b \otimes \xi_\beta \nonumber \\
& & + \epsilon \sum_{(c\gamma)\not=(b\beta)} \frac{V_{c\gamma,b\beta}}{\mu_b-\mu_c + \nu_\beta-\nu_\gamma} \zeta_c \otimes \xi_\gamma \nonumber \\
& & + \epsilon^2 \sum_{\begin{array}{c} \scriptstyle (d\delta)\not=(b\beta) \\ \scriptstyle  (c\gamma)\not=(b\beta) \end{array}} \frac{V_{d\delta,c\gamma}V_{c\gamma,b\beta}-V_{c\gamma,b\beta}V_{b\beta,b\beta}}{(\mu_b-\mu_d + \nu_\beta-\nu_\delta)(\mu_b-\mu_c+\nu_\beta-\nu_\gamma)} \zeta_d \otimes \xi_\delta \nonumber \\
& & + \mathcal O(\epsilon^3)
\end{eqnarray}
is not normalized (at the order $\epsilon^3$). This induces some difficulties to define an adiabatic transport formula, especially for the definition of the generator of the geometric phase. A normalization factor could be very complicated and difficult to use. We prefer to use a biorthonormal basis $\{\phi_{b\beta}^*\}_{b\beta}$ defined such that
\begin{equation}
\llangle \phi_{c\gamma}^*|\phi_{b\beta} \rrangle = \delta_{cb} \delta_{\gamma\beta} + \mathcal O(\epsilon^3)
\end{equation}
Such an approach is able to define a correct geometric phase generator \cite{Viennot7,Viennot8}. In the present context, the biorthonormal eigenvectors are
\begin{equation}
\llangle \phi_{b\beta}^*| = \llangle \phi_{b\beta}| - \epsilon^2 \sum_{(c\gamma)} X_{b\beta,c\gamma} \llangle \zeta_c \otimes \xi_\gamma|
\end{equation}
with
\begin{eqnarray}
& & X_{b\beta,c\gamma} \nonumber \\
& & = \sum_{(d\delta)\not=(c\gamma)} \frac{V_{b\beta,d\delta}V_{d\delta,c\gamma}-V_{d\delta,c\gamma}V_{c\gamma,c\gamma}}{(\mu_c-\mu_b+\nu_\gamma-\nu_\beta)(\mu_c-\mu_d+\nu_\gamma-\nu_\delta)}(1-\delta_{bc}\delta_{\beta\gamma}) \nonumber \\
& & + \sum_{(d\delta)\not=(b\beta)} \frac{V_{d\delta,c\gamma}V_{b\beta,d\delta}-V_{b\beta,d\delta}V_{b\beta,b\beta}}{(\mu_b-\mu_c+\nu_\beta-\nu_\gamma)(\mu_b-\mu_d+\nu_\beta-\nu_\delta)}(1-\delta_{bc}\delta_{\beta\gamma}) \nonumber \\
& & + \sum_{(d\delta)\not=(b\beta)} \frac{V_{d\delta,c\gamma}V_{b\beta,d\delta}}{(\mu_c-\mu_d+\nu_\gamma-\nu_\delta)(\mu_b-\mu_d+\nu_\beta-\nu_\delta)}(1-\delta_{dc}\delta_{\delta\gamma})
\end{eqnarray}

\begin{propo}
In the conditions of the weak adiabatic theorem (theorem \ref{weakth}) we have $\forall s \in[0,1]$
\begin{eqnarray}
\rho(s) & = & \Ad\left[\mathcal U_\alpha(s)\right] \rho_{a\alpha}^{(2)}(s) \nonumber \\
& & \qquad + \epsilon^2 \sum_{\delta\not=\alpha} \Ad\left[ \mathcal W_{\delta \alpha}(s) \mathcal U_\alpha(s) \right] |\zeta_a(s)\rangle\langle \zeta_a(s)|
\nonumber \\
& & \qquad + \mathcal O(\max(\frac{1}{T},\epsilon^3)) \label{adiabtransp2}
\end{eqnarray}
\begin{eqnarray}
\mathcal U_\alpha(s) & = & \Teg^{-\ihbar^{-1} T \int_0^s E_\alpha^{(2)}(\sigma)d\sigma} \Ted^{- \int_0^s A^{(2)}_\alpha (\sigma)d\sigma} \\
\mathcal W_{\delta \alpha}(s) & = & \sum_d \sum_{c\not=d} \frac{V_{d\delta,c\alpha}(s)}{\Delta_{c\alpha,d\delta}(s)} |\zeta_d(s)\rangle\langle \zeta_c(s)|
\end{eqnarray}
with $\Delta_{c\alpha,d\delta} = \mu_c-\mu_d +\nu_\alpha-\nu_\delta+\epsilon V_{c\alpha,c\alpha}$. The first order dynamical phase generator is defined as being
\begin{equation}
E_\alpha^{(2)}(s) = \sum_{b,c} \left(\lambda_{b\alpha}(s) \delta_{bc} - \frac{\ihbar}{T} \eta_{\alpha bc}^{(2)}(s)\right) |\zeta_{b\alpha}^{(2)}(s) \rangle \langle \zeta_{c\alpha}^{*(2)}(s)| \in \mathcal L(\mathcal H_\S)
\end{equation}
and the first order geometric phase generator is defined as being
\begin{eqnarray}
A^{(2)}_\alpha(s) & = & \sum_{b,c} \left(\langle \zeta_{b\alpha}^{*(2)}(s)|\zeta_{c\alpha}^{(2)\prime}(s) \rangle \right. \nonumber \\
& & \quad  + \epsilon^2 \sum_{\begin{array}{c} \scriptstyle d\not=c \\ \scriptstyle  f\not=b \\ \scriptstyle \delta\not=\alpha \end{array}} \left. \frac{V_{d\delta,c\alpha}(s)V_{b\alpha,f\delta}(s)}{\Delta_{c\alpha,d\delta}(s)\Delta_{b\alpha,f\delta}(s)} \langle \zeta_f(s)|\zeta'_d(s)\rangle \right) \nonumber \\
& & \qquad \qquad \times |\zeta_{b\alpha}^{(2)}(s) \rangle \langle \zeta_{c\alpha}^{*(2)}(s)| \in \mathcal L(\mathcal H_\S)
\end{eqnarray}
with
\begin{eqnarray}
\zeta_{b\alpha}^{(2)} & = & \zeta_b + \epsilon \sum_{d\not=b} \frac{V_{d\alpha,b\alpha}}{\Delta_{b\alpha,d\alpha}} \zeta_d \nonumber \\
& & \quad + \epsilon^2 \sum_{\begin{array}{c} \scriptstyle d\not=b \\ \scriptstyle  e\not=b \end{array}} \frac{V_{d\alpha,e\alpha}V_{e\alpha,b\alpha}-V_{e\alpha,b\alpha}V_{b\alpha,b\alpha}}{\Delta_{b\alpha,d\alpha}\Delta_{b\alpha,e\alpha}} \zeta_d
\end{eqnarray}
\begin{eqnarray}
\langle \zeta_{c\alpha}^{*(2)}| & = & \langle \zeta_c| \nonumber \\
& & + \epsilon \sum_{f\not=c} \frac{V_{c\alpha,f\alpha}}{\Delta_{c\alpha,f\alpha}} \langle\zeta_f| \nonumber \\
& & + \epsilon^2 \sum_{\begin{array}{c} \scriptstyle d\not=c \\ \scriptstyle  e\not=c \end{array}} \frac{V_{e\alpha,d\alpha}V_{c\alpha,e\alpha}-V_{c\alpha,e\alpha}V_{c\alpha,c\alpha}}{\Delta_{c\alpha,d\alpha}\Delta_{c\alpha,e\alpha}} \langle \zeta_d| \nonumber \\
& & - \epsilon^2 \sum_{\begin{array}{c} \scriptstyle e \\ \scriptstyle  k\not=e \end{array}} \frac{(V_{c\alpha,k\alpha}V_{k\alpha,e\alpha}-V_{k\alpha,e\alpha}V_{e\alpha,e\alpha})(1-\delta_{ce})}{\Delta_{e\alpha,c\alpha}\Delta_{e\alpha,k\alpha}} \langle \zeta_e| \nonumber \\
& &  - \epsilon^2 \sum_{\begin{array}{c} \scriptstyle e \\ \scriptstyle  k\not=c \end{array}} \frac{(V_{k\alpha,e\alpha}V_{c\alpha,k\alpha}-V_{c\alpha,k\alpha}V_{c\alpha,c\alpha})(1-\delta_{ce})}{\Delta_{c\alpha,e\alpha}\Delta_{c\alpha,k\alpha}} \langle \zeta_e| \nonumber \\
& &  - \epsilon^2 \sum_{\begin{array}{c} \scriptstyle e \\ \scriptstyle  k\not=c \\ \scriptstyle  \gamma\not=\alpha \end{array}} \frac{V_{k\gamma,e\alpha}V_{c\alpha,k\gamma}(1-\delta_{ke})}{\Delta_{e\alpha,k\gamma}\Delta_{c\alpha,k\gamma}} \langle \zeta_e|
\end{eqnarray}

\begin{eqnarray}
\eta^{(2)}_{\alpha bc} & = & \langle \xi_\alpha|\xi_\alpha'\rangle \delta_{bc} \nonumber \\
& & + \epsilon \underset{\delta \not=\alpha}{\sum} \frac{V_{b\delta,c\alpha} \langle\xi_\alpha|\xi_\delta'\rangle (1-\delta_{bc})}{\Delta_{c\alpha,b\delta}} \nonumber \\
& & + \epsilon \underset{\gamma \not=\alpha}{\sum} \frac{V_{b\alpha,c\gamma} \langle\xi_\gamma|\xi_\alpha'\rangle (1-\delta_{bc})}{\Delta_{b\alpha,c\gamma}} \nonumber \\
& & + \epsilon^2 \underset{\begin{array}{c}  \scriptstyle d\not=c \\  \scriptstyle  \delta\not=\alpha \\  \scriptstyle  \gamma \end{array}}{\sum} \frac{V_{d\delta,c\alpha} V_{b\alpha,d\gamma}(1-\delta_{db})\langle \xi_\gamma|\xi_\delta'\rangle}{\Delta_{c\alpha,d\delta}\Delta_{b\alpha,d\gamma}} \nonumber \\
& &\ + \epsilon^2 \underset{\begin{array}{c}  \scriptstyle \delta\not=\alpha \\  \scriptstyle  e\not= c \\  \scriptstyle  \phi \not= a \end{array}}{\sum} \frac{(V_{b\delta,e\phi}V_{e\phi,c\alpha}-V_{e\phi,c\alpha}V_{c\alpha,c\alpha})(1-\delta_{bc})\langle \xi_\alpha|\xi_\delta'\rangle}{\Delta_{c\alpha,b\delta}\Delta_{c\alpha,e\phi}} \nonumber \\
& &+ \epsilon^2 \underset{\begin{array}{c} \scriptstyle \gamma\not=\alpha \\ \scriptstyle  d\not= c \end{array}}{\sum} \frac{V_{b\alpha,d\gamma}V_{d\alpha,c\alpha}(1-\delta_{db}) \langle \xi_\gamma|\xi_\alpha'\rangle}{\Delta_{b\alpha,d\gamma}\Delta_{d\alpha,c\alpha}} \nonumber \\
& & +\epsilon^2 \underset{\begin{array}{c}  \scriptstyle \delta\not=\alpha \\  \scriptstyle  e\not= b \\  \scriptstyle  \phi \not= \alpha \end{array}}{\sum} \frac{(V_{e\phi,c\delta}V_{b\alpha,e\phi}-V_{b\alpha,e\phi}V_{b\alpha,b\alpha})(1-\delta_{bc})\langle \xi_\delta|\xi_\alpha'\rangle}{\Delta_{b\alpha,c\delta}\Delta_{b\alpha,e\phi}} \nonumber \\
& & -\epsilon^2 \underset{\begin{array}{c}  \scriptstyle \phi\not=\alpha \\  \scriptstyle  k\not= c \\  \scriptstyle  \gamma \not= \phi,\alpha \end{array}}{\sum} \frac{(V_{b\alpha,k\gamma}V_{k\gamma,c\phi} - V_{k\gamma,c\phi} V_{c\phi,c\phi})(1-\delta_{(b\alpha),(c\phi)}) \langle \xi_\phi|\xi_\alpha'\rangle}{\Delta_{c\phi,b\alpha}\Delta_{c\phi,k\gamma}} \nonumber \\
& & -\epsilon^2 \underset{\begin{array}{c} \scriptstyle \phi\not=\alpha \\  \scriptstyle  k\not= b \\  \scriptstyle  \gamma \not= \alpha \end{array}}{\sum} \frac{(V_{k\gamma,c\phi}V_{b\alpha,k\gamma}-V_{b\alpha,k\gamma}V_{b\alpha,b\alpha})(1-\delta_{(b\alpha),(c\phi)}) \langle \xi_\phi|\xi_\alpha' \rangle}{\Delta_{b\alpha,c\phi}\Delta_{b\alpha,k\gamma}} \nonumber \\
& &-\epsilon^2 \underset{\begin{array}{c} \scriptstyle \phi\not=\alpha \\  \scriptstyle  k\not= b \\  \scriptstyle  \gamma \not= \alpha \end{array}}{\sum} \frac{V_{k\gamma,c\phi} V_{b\alpha,k\gamma}(1-\delta_{(k\gamma),(c\phi)}) \langle \xi_\phi|\xi_\alpha'\rangle}{\Delta_{c\phi,k\gamma}\Delta_{b\alpha,k\gamma}}
\end{eqnarray}
$\rho^{(2)}_{a\alpha}$ is the corrected density matrix defined as
\begin{eqnarray}
\rho^{(2)}_{a\alpha}& = &  \rho_{a\alpha} \nonumber \\
& & \quad + \epsilon^2 \sum_{\begin{array}{c} \scriptstyle c \\ \scriptstyle d\not=c \\ \scriptstyle  \gamma\not=\alpha \end{array}} \frac{V_{c\alpha,d\gamma}V_{d\gamma,a\alpha}(1-\delta_{ad})}{\Delta_{c\alpha,d\gamma}\Delta_{a\alpha,d\gamma}} |\zeta_c \rangle \langle \zeta_a| \nonumber\\
& & \quad + \epsilon^2 \sum_{\begin{array}{c} \scriptstyle f \\ \scriptstyle d\not=c \\ \scriptstyle  \gamma\not=\alpha \end{array}} \frac{V_{d\gamma,f\alpha}V_{a\alpha,d\gamma}(1-\delta_{ad})}{\Delta_{f\alpha,d\gamma}\Delta_{a\alpha,d\gamma}} |\zeta_a \rangle \langle \zeta_f| \nonumber\\
& & \quad - \epsilon^2 \sum_{\begin{array}{c} \scriptstyle c,f \\ \scriptstyle  \delta\not=\alpha \end{array}} \frac{V_{a\alpha,f\delta}V_{c\delta,a\alpha}(1-\delta_{ac})(1-\delta_{af})}{\Delta_{a\alpha,f\delta}\Delta_{a\alpha,c\delta}} |\zeta_c \rangle \langle \zeta_f|
\end{eqnarray}
\end{propo}

The proof is very long but its development is very similar to the first order case except that we need to take into account the biorthonormality and that some second order extra terms involving indexes $\delta \not= \alpha$ of $\mathcal E$, which are not killed by the partial trace. The significance of these extra terms and of the higher complexity of the adiabatic transport formula, are discussed in the following section.

\subsection{Discussion about the operator-valued phases}
Operator-valued geometric phases have been introduced in \cite{Andersson,Viennot1,Viennot2} for density matrices. We recall rapidly the motivation of such geometric phases. The quantum control problems are characterized by the condition $H(1)=H(0)$ (we start and we end with control system off). This induces that $\phi_{a\alpha}(1) = \phi_{a\alpha}(0)$ and then $\rho_{a\alpha}(1) = \rho_{a\alpha}(0)$. But to solve a quantum control problem, we need that $\|\rho(1) - \rho_{target}\|$ to be minimal (with $\rho_{target}$ the control goal and $\rho(s)$ the density matrix of the dynamics such that $\rho(0) = \rho_{a\alpha}(0)$). But if $\rho(1) = \rho_{a\alpha}(1)$ (strong adiabatic regime) it is impossible to solve a quantum control problem by an adiabatic scheme (unless the initial condition is already the control target). In adiabatic quantum control, it needs that $\rho(1) = U \rho_{a\alpha}(1) U^\dagger$ with $U$ an operator of $\S$ associated with the adiabatic transport of the mixed state $\rho_{a\alpha}$, and transforming $\rho_{a\alpha}(1)$ such that $\rho(1)$ is close to $\rho_{target}$. In comparison with the adiabatic transport of pure states of closed systems ($\psi(1) = e^{\imath \varphi} \phi_a(1)$, where $\phi_a$ is an instantaneous eigenvector and $e^{\imath \varphi}$ is the product of a dynamical and a geometric phases); $U$ plays the role of the product of a dynamical phase and a geometric phase. But these phases are operator-valued since $U$ is an operator. This is well what we find with the adiabatic transport formula of $\rho_{a\alpha}$ in the weak adiabatic regime.\\ 

In \cite{Viennot1} by an analysis based on a generalization of the geometric structure describing the usual adiabatic geometric phases (using a non-commutative Hilbert space -- a $C^*$-module -- and a categorical principal bundle) the generator of an operator-valued geometric phase has been defined by (we use the present notations)
\begin{equation}
\label{Adefgene}
A_\alpha = \tr_\E\left(|P_{\bullet \alpha} \phi_{a\alpha}'\rrangle \llangle \phi_{a\alpha}|\right) \rho_{a\alpha}^{-1}
\end{equation}
where $\rho_{a\alpha}^{-1}$ is the pseudo-inverse of $\rho_{a\alpha}$ ($\rho_{a\alpha} \rho_{a\alpha}^{-1} = 1 - P_{\ker \rho_{a\alpha}}$ where $P_{\ker \rho_{a\alpha}}$ is the orthogonal projection onto the kernel of $\rho_{a\alpha}$). By considering the perturbative expansions  we find that
\begin{eqnarray}
A_\alpha & = & A^{(0)} + \langle \xi_\alpha|\xi_\alpha'\rangle  + \mathcal O(\epsilon) \\
& = & A^{(1)}_{\alpha}  + \langle \xi_\alpha|\xi_\alpha'\rangle + \mathcal O(\epsilon^2)
\end{eqnarray}
Up to a $U(1)$-gauge change leaving invariant the density matrix $\rho(s)$ ($e^{-\int_0^s \langle \xi_\alpha|\xi_\alpha'\rangle d\sigma} \in U(1)$), the operator valued geometric phases found in the present paper coincides with the definition introduced in \cite{Viennot1} which is a generalization of the geometric phases introduced in \cite{Sjoqvist,Tong,Dajka,Andersson}.\\

The role of the operator-valued dynamical phase is interesting. Suppose temporarily that $\zeta_{b\alpha}^{(1)}$ is constant (independent of $s$) but not $\lambda_{b\alpha}(s)$. In this assumption we have
\begin{equation}
\rho(s) \simeq \Ad\left[\Teg^{-\ihbar^{-1} T \int_0^s E_\alpha^{(1)}(\sigma) d\sigma}\right] \rho_{a\alpha}
\end{equation}
This induces that
\begin{equation}
\ihbar \dot \rho \simeq \left[E_\alpha^{(1)},\rho \right]
\end{equation}
This expression is very similar to the Liouville-von Neumann equation of an isolated system (\cite{Breuer}): if $\S$ is isolated and governed by the self-adjoint Hamiltonian $H_\S \in \mathcal L(\mathcal H_\S)$, we have
\begin{equation}
\ihbar \dot \rho = \left[H_\S,\rho \right]
\end{equation}
$E_\alpha^{(1)}$ plays then the role of an effective Hamiltonian of $\S$ taking into account effects induced by $\E$. We can consider $E_\alpha^{(1)}$ as the effective Hamiltonian of $\S$ dressed by $\E$ like the Floquet Hamiltonian of an atom interacting with a strong laser field is the effective Hamiltonian of the atom dressed by the photons \cite{Viennot6,Guerin2}.\\
In the reality $\zeta_{b\alpha}^{(1)}$ depends on the reduced time $s$, and $E_\alpha^{(1)}$ is described by using a moving basis. The operator-valued geometric phase (as all geometric phases) is just a correction to take into account the movement of the basis (like for the simpler example, the inertial forces are corrections in Newtonian mechanics to take into account a description in a non-inertial frame).\\

Concerning the second order adiabatic transport formula, we suppose temporarily again that $\zeta_{b\alpha}^{(2)}$ and $\mathcal W_{\delta \alpha}$ are constant. By using the expression \ref{adiabtransp2}, $\rho(s) \simeq \Ad\left[\mathcal U_\alpha\right] \rho_{a\alpha}^{(2)} + \epsilon^2 \sum_{\delta \not=\alpha} \Ad\left[\mathcal W_{\delta \alpha} \mathcal U_\alpha \right] \rho_{a\alpha}^{(2)}$ satisfies
\begin{eqnarray}
\ihbar \dot \rho & \simeq & E_\alpha^{(2)} \rho - \rho E_\alpha^{(2) \dagger} + \epsilon^2 \sum_{\delta\not=\alpha} \mathcal W_{\delta \alpha} (E_\alpha^{(2)} \rho-\rho E_\alpha^{(2)\dagger}) \mathcal W_{\delta \alpha}^\dagger \\
& \simeq & [E_{\alpha+}^{(2)},\rho] + \imath \{E_{\alpha-}^{(2)},\rho\} \nonumber \\
& & \qquad + \epsilon^2 \sum_{\delta\not=\alpha} \mathcal W_{\delta \alpha} ([E_{\alpha+}^{(2)},\rho] + \imath \{E_{\alpha-}^{(2)},\rho\} ) \mathcal W_{\delta \alpha}^\dagger
\end{eqnarray}
where $E_{\alpha+}^{(2)} = \frac{1}{2}(E_\alpha^{(2)}+E_\alpha^{(2)\dagger})$ and $E_{\alpha-}^{(2)} = \frac{1}{2\imath}(E_\alpha^{(2)}-E_\alpha^{(2)\dagger})$; the braces denote the anticommutator ($\{A,B\} = AB+BA$). The dynamical phase generator has the form $E_\alpha^{(2)} = E_{\alpha 0}^{(2)} - \epsilon^2 \sum_{\delta\not=\alpha} E_{\alpha 0}^{(2)} \mathcal W_{\delta \alpha}^\dagger \mathcal W_{\delta \alpha}$, and then $E_{\alpha+}^{(2)} = E_{\alpha 0}^{(2)} - \frac{\epsilon^2}{2} \sum_{\delta\not= \alpha} \{E_{\alpha 0}^{(2)},\mathcal W_{\delta \alpha}^\dagger \mathcal W_{\delta \alpha}\}$ and $E_{\alpha-}^{(2)} =  - \frac{\epsilon^2}{2 \imath} \sum_{\delta\not= \alpha} [E_{\alpha 0}^{(2)},\mathcal W_{\delta \alpha}^\dagger \mathcal W_{\delta \alpha}]$. This implies that
\begin{eqnarray}
\ihbar \dot \rho & \simeq & [E_{\alpha+}^{(2)},\rho] - \frac{\epsilon^2}{2} \sum_{\delta \not= \alpha} \{[E_{\alpha 0}^{(2)},\mathcal W_{\delta \alpha}^\dagger \mathcal W_{\delta \alpha}],\rho\} \nonumber \\
& & \qquad + \epsilon^2 \sum_{\delta\not= \alpha} \mathcal W_{\delta \alpha} [E_{\alpha 0}^{(2)},\rho ] \mathcal W_{\delta \alpha}^\dagger
\end{eqnarray}
Let $\Gamma_{\alpha \delta 0} = \mathcal W_{\delta \alpha} + \imath \mathcal W_{\delta \alpha} E_{\alpha 0}^{(2)}$, $\Gamma_{\alpha \delta 1} = \mathcal W_{\delta \alpha}$ and $\Gamma_{\alpha \delta 2} =  \mathcal W_{\delta \alpha} E_{\alpha 0}^{(2)}$, we have then
\begin{eqnarray}
\ihbar \dot \rho & \simeq & [E_{\alpha+}^{(2)},\rho] - \frac{\imath \epsilon^2}{2} \sum_{\delta \not= \alpha}  \sum_k \gamma^k \{\Gamma_{\alpha \delta k}^\dagger \Gamma_{\alpha \delta k},\rho\} \nonumber \\
& & \qquad + \imath \epsilon^2 \sum_{\delta\not= \alpha} \sum_k\gamma^k \Gamma_{\alpha \delta k} \rho \Gamma_{\alpha \delta k}^\dagger
\end{eqnarray}
with $\gamma^0 = 1$ and $\gamma^1=\gamma^2=-1$. This last equation is similar to the Lindblad equation of an open quantum system in Markovian approximation \cite{Breuer} (except that in the strict Lindblad theory $\gamma^k>0$ for all $k$). $E_\alpha^{(2)}$ and $\mathcal W_{\alpha \delta}$ generate then an effective Lindblad equation for $\mathcal S$ in contact with $\mathcal E$. The extra terms involving indexes of $\mathcal E$ different from $\alpha$ in equation \ref{adiabtransp2}, are then associated with the ``quantum jumps'' (see \cite{Breuer}). The geometric phase is anew a correction to take into account that the biorthonormal basis $\{\zeta_{b\alpha}^{(2)},\zeta_{b\alpha}^{*(2)}\}_b$ is moving.

\subsection{The thermal bath case}
When $\E$ is a large subsystem, it can be interesting to consider it at $s=0$ as being a thermal bath, i.e. $\E$ is described by the density matrix
\begin{equation}
\rho_B = \frac{e^{-\underline \beta H_\E(0)}}{Z} = \sum_\alpha \frac{e^{-\underline \beta \nu_\alpha(0)}}{Z} |\xi_\alpha(0)\rangle \langle\xi_\alpha(0)|
\end{equation}
where $\underline \beta = \frac{1}{k_B \underline T}$ ($\underline T$ being the temperature of the bath and $k_B$ being the Boltzmann constant, the underline is just a notation to avoid confusions with state indexes or with the duration of the evolution). The partition function is $Z = \tr_\E e^{-\underline \beta H_\E(0)}$. Let $\rho_{\mathcal U} \in \mathcal L(\mathcal H_\S \otimes \mathcal H_\E)$ be the density matrix of the complete bipartite system solution of the Liouville-von Neumann equation:
\begin{equation}
\frac{\ihbar}{T} \rho_{\mathcal U}'(s) = \left[H_\S(s)\otimes 1_\E + 1_\S \otimes H_\E(s) + \epsilon V(s),\rho_{\mathcal U}(s)\right]
\end{equation}
\begin{equation}
\rho_{\mathcal U}(0) = \sum_\alpha \frac{e^{- \underline \beta \nu_\alpha(0)}}{Z} |\phi_{a\alpha}(0) \rrangle \llangle \phi_{a\alpha}(0)|
\end{equation}
We have $\tr_\S \rho_{\mathcal U}(0) = \rho_B + \mathcal O(\epsilon^2)$ implying that $\E$ is well a thermal bath (moreover $\rho_{\mathcal U}$ is a steady state of the $H(0)$). The solution of the Liouville-von Neumann equation  is
\begin{equation}
\rho_{\mathcal U}(s) = \sum_\alpha \frac{e^{- \underline \beta \nu_\alpha(0)}}{Z} |\psi_{(a\alpha)}(s) \rrangle \llangle \psi_{(a\alpha)}(s)|
\end{equation}
where $\psi_{(a\alpha)}$ is the solution of the Schr\"odinger equation $\frac{\ihbar}{T} \psi_{(a\alpha)}' = H\psi_{(a\alpha)}$ with the initial condition $\psi_{(a\alpha)}(0) = \phi_{a\alpha}(0)$. At the weak adiabatic limit we have then
\begin{eqnarray}
\rho(s) & = & \tr_\E \rho_{\mathcal U}(s) \\
& = & \sum_\alpha \frac{e^{- \underline \beta \nu_\alpha(0)}}{Z} \Ad\left[\Teg^{-\ihbar^{-1} T \int_0^s E_\alpha^{(1)}(\sigma)d\sigma} \Ted^{- \int_0^s A^{(1)}_\alpha (\sigma)d\sigma} \right] \rho_{a\alpha}(s) \nonumber \\
& & \qquad + \mathcal O(\max(\frac{1}{T},\epsilon^2))
\end{eqnarray}

\subsection{Eigenvalue crossings of $\E$ in the weak adiabatic regime} \label{crossing}
The weak adiabatic theorem (theorem \ref{weakth}) requires that the eigenvalue of $\E$, $\nu_\alpha(s)$, does not cross another eigenvalue. If this requirement is natural for the control problem, it can be not realized in the practice. Suppose that $\exists s_* \in [0,1]$ such that $\nu_\alpha(s_*) = \nu_\beta(s_*)$ (no other crossings implying $\nu_\alpha$ and $\nu_\beta$ occur). We suppose that the conditions of the weak adiabatic theorem are satisfied except in the neighbourhood of $s_*$. Due to the nonadiabatic transitions induced in the neighbourhood of $s_*$ by this crossing, the density matrix becomes (for $s\gg s_*$):
\begin{eqnarray}
\rho(s) & = & (1-p) \Ad\left[\Teg^{-\ihbar^{-1}T\int_0^s E_\alpha^{(1)} d\sigma} \Ted^{-\int_0^s A_\alpha^{(1)}d\sigma}\right] \rho_{a\alpha}(s) \nonumber \\
& & + p  \Ad\left[\Teg^{-\ihbar^{-1}T\int_0^s E_\beta^{(1)} d\sigma} \Ted^{-\int_0^s A_\beta^{(1)}d\sigma}\right] \rho_{a\beta}(s) \nonumber \\
& & + \sqrt{(1-p)p}e^{\imath \varphi} \underset{\alpha \leftrightarrow \beta}{\Ad} \left[\Teg^{-\ihbar^{-1}T\int_0^s E^{(1)}_\bullet d\sigma} \Ted^{-\int_0^s A^{(1)}_\bullet d\sigma} \right] \tau_{a\alpha\beta}(s)  \nonumber \\
& & + \sqrt{(1-p)p}e^{-\imath \varphi} \underset{\beta \leftrightarrow \alpha}{\Ad} \left[\Teg^{-\ihbar^{-1}T\int_0^s E^{(1)}_\bullet d\sigma} \Ted^{-\int_0^s A^{(1)}_\bullet d\sigma} \right] \tau_{a\beta\alpha}(s)\nonumber \\
& & + \mathcal O(\max(\frac{1}{T},\epsilon^2))
\end{eqnarray}
where $\underset{\alpha \leftrightarrow \beta}{\Ad}[U_\bullet]\tau = U_\alpha \tau U_\beta^\dagger$ and 
\begin{eqnarray}
\tau_{a\alpha\beta} & = & \tr_\E |\phi_{a\alpha} \rrangle \llangle \phi_{a\beta}|\\
& = & \epsilon \sum_{d\not=a} \frac{V_{d\beta,a\alpha}}{\mu_a-\mu_d+\nu_\alpha-\nu_\beta+\epsilon V_{a\alpha,a\alpha}} |\zeta_d\rangle \langle \zeta_a| \nonumber \\
& & \quad + \epsilon \sum_{d\not=a} \frac{V_{a\beta,d\alpha}}{\mu_a-\mu_d+\nu_\beta-\nu_\alpha+\epsilon V_{a\beta,a\beta}} |\zeta_a\rangle \langle \zeta_d| + \mathcal O(\epsilon^2)
\end{eqnarray}
$p$ is the probability of the nonadiabatic transition from $\xi_\alpha$ to $\xi_\beta$ induced by the passage through the crossing, and $\varphi$ is a phase difference accumulated during the nonadiabatic transition. It is clear that the crossing of eigenenergies of $\E$ generates a decoherence effect in the density matrix of $\S$ that we call \textit{kinematic decoherence} since it is induced by the variation of the control system with respect to the time. In practice it can be difficult to compute explicitly $p$, but if we suppose that $\forall s \in \mathcal V(s_*)$, $\nu_\beta(s)-\nu_\alpha(s) = \aleph (s-s_*)$ ($\aleph$ being a constant) and that $V_{a\alpha,a\beta}$ is independent of $s$, then $p$ can be estimated by the Landau-Zener formula \cite{Landau,Zener}, i.e. $p= e^{-2\pi \frac{T\epsilon^2 V_{a\alpha,a\beta}^2}{\hbar |\aleph|}}$.

\section{Examples}
In this section we present two examples of bipartite quantum systems and we study their adiabatic dynamics. We want to compare their real dynamics (numerically computed by using a split operator method without another approximation), to the prediction of the usual adiabatic transport formula for $\S$ alone (by neglecting the influence of $\E$, an approximation currently considered in adiabatic control methods), and to the prediction of the adiabatic transport formulae with operator valued phases which considers $\S$ dressed by states of $\E$ (the operator-valued phases are numerically computed by the same split operator method with the same time discretisation, nevertheless the dimensions of the matrices -- $\dim \mathcal H_\S$ -- is reduced in comparison with the ``exact'' computation -- $\dim \mathcal H_\S \times \dim \mathcal H_\E$ --).

\subsection{Control of atomic qubits}
\subsubsection{The model:}
We consider a two level atom $\S$ interacting with a laser field which is governed in the rotating wave approximation with one photon by the Hamiltonian
\begin{eqnarray}
H_\S(s) & = & \frac{\hbar}{2} \left( \begin{array}{cc} 0 & \Omega(s) e^{\imath \varphi(s)} \\ \Omega(s) e^{-\imath \varphi(s)} & 2 \Delta(s) \end{array} \right) \\
& = & \frac{\hbar}{2} \left( \Omega(s) \cos \varphi(s) \sigma_x + \Omega(s) \sin \varphi(s) \sigma_y + \Delta(s) (\id-\sigma_z) \right)
\end{eqnarray}
where $\Omega(s)$ is the product between the electric field strength and the dipolar moment of the atom, $\varphi$ is the dephasing of the laser field, and $\Delta$ is the detuning (the energy gap between the two atomic states minus the energy of one photon of the laser field). This system can be viewed like a model of one qubit where the laser field is the control system performing a one input/output logic gate on it. $(\sigma_x,\sigma_y,\sigma_z)$ are the Pauli matrices.\\
A second atom (qubit) $\E$ is in contact with the first one and is governed by the following Hamiltonian
\begin{equation}
H_\E = \left(\begin{array}{cc} 0 & 0 \\ 0 & \hbar\omega_e \end{array} \right) = \frac{\hbar\omega_e}{2} (\id-\sigma_z)
\end{equation}
The interaction between the both atoms is chosen as being
\begin{eqnarray}
\epsilon V & = & \epsilon \left(\begin{array}{cccc} V_0 & V_1 & 0 & V_3 \\ V_1 & V_0 & V_3 & 0 \\ 0 & V_3 & 2V_0 & V_2 \\ V_3 & 0 & V_2 & 2V_0 \end{array} \right) \\
& = & \epsilon \left( (\frac{V_0}{2} \id \otimes (\id + \sigma_z) + V_0 \id \otimes (\id - \sigma_z) \right. \nonumber \\
& & \quad \left. + \frac{V_1}{2} \sigma_x \otimes (\id + \sigma_z) +  \frac{V_2}{2} \sigma_x \otimes (\id - \sigma_z) + V_3 \sigma_x \otimes \sigma_x \right)
\end{eqnarray}
in a matrix representation where the two first inputs are associated with the both states of $\S$ and the ground state of $\E$, and the two last inputs are associated with the both states of $\S$ and the excited state of $\E$. $\epsilon \ll 1$ is the perturbative parameter.\\
Let $r(s) = \sqrt{\Omega(s)^2+\Delta(s)^2}$ and $\theta(s) = \arctan \frac{\Omega(s)}{\Delta(s)}$ be variable changes of the control parameters. The eigenvalues and the eigenvectors of the both components of the bipartite system are
\begin{eqnarray}
\mu_0(s) =  \frac{\hbar}{2} r(s) (\cos \theta(s) -1) & \qquad & \zeta_0(s) = \left(\begin{array}{c} - \cos \frac{\theta(s)}{2} \\ e^{-\imath \varphi(s)} \sin \frac{\theta(s)}{2} \end{array} \right) \\
\mu_1(s)  =  \frac{\hbar}{2} r(s) (\cos \theta(s) +1) & \qquad & \zeta_1(s) = \left(\begin{array}{c} e^{\imath \varphi(s)} \sin \frac{\theta(s)}{2} \\  \cos \frac{\theta(s)}{2} \end{array} \right) \\
\nu_0  =  0 & \qquad & \xi_0 = \left(\begin{array}{c} 1 \\ 0 \end{array} \right) \\
\nu_1  =  \hbar\omega_e & \qquad & \xi_1 = \left(\begin{array}{c} 0 \\ 1 \end{array} \right)
\end{eqnarray}
The control is fixed by the following variation of the control parameters:
\begin{eqnarray}
r(s) & = & r_{max} + (r_{min}-r_{max}) e^{-25 (s-0.5)^2} \\
\theta(s) & = & \theta_{max} \sin (\pi s) \\
\varphi(s) & = & 2\pi s
\end{eqnarray}
corresponding to laser pulses and a laser frequency modulation represented figure \ref{lasermod} and with a drifting phase.
\begin{figure}
\begin{center}
\includegraphics[width=8cm]{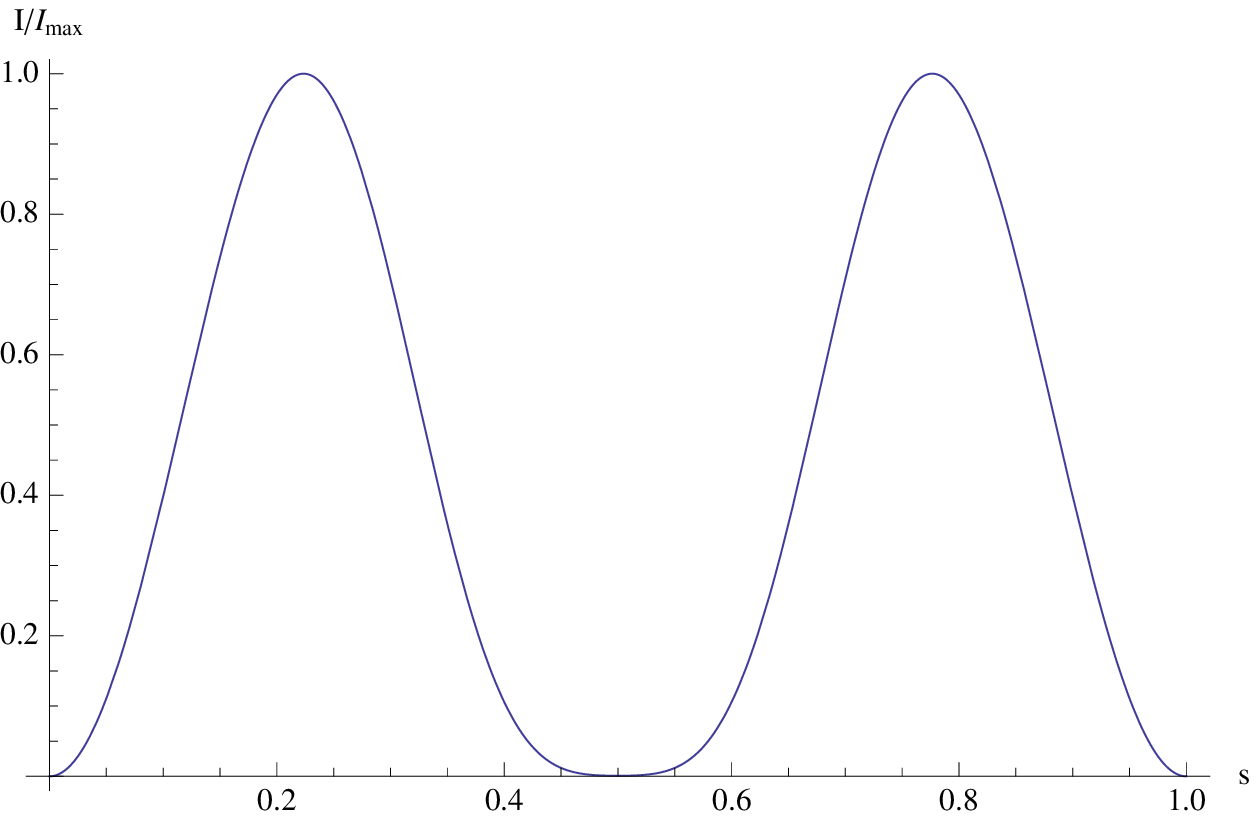} \includegraphics[width=8cm]{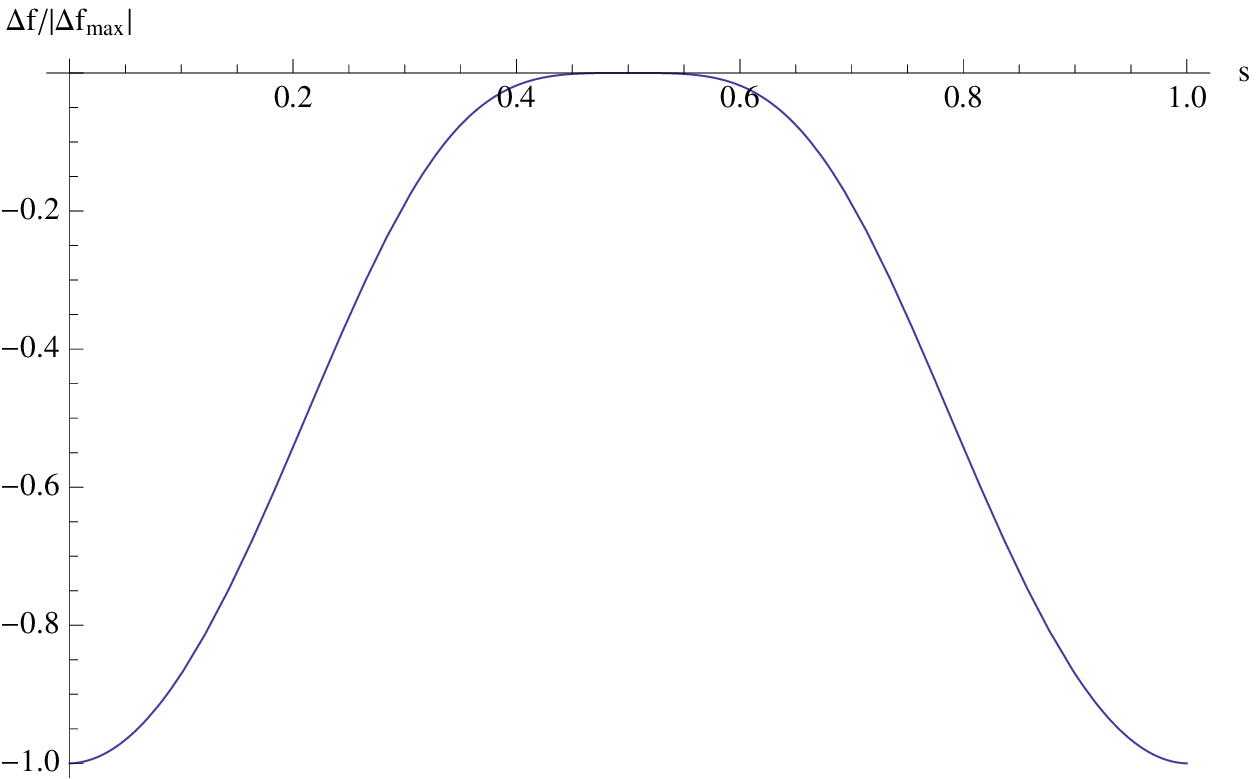}
\caption{\label{lasermod} Intensity of the laser pulses applied on the atom with respect to the reduced time (up) and difference of the modulated laser frequency with the frequency of the atomic transition with respect to the reduced time (down).}
\end{center}
\end{figure}

\subsubsection{Adiabatic transports:}
We start with the both qubits in the ground state, $\phi_{00}(0)$ with $\Omega(0)=0$ (the control laser is off). The adiabatic transport of the density matrix for $\S$ alone is
\begin{equation}
\rho_{alone-ad}(s) = |\zeta_{0}(s) \rangle \langle \zeta_{0}(s)|
\end{equation}
If the dynamics of the both qubits is strongly adiabatic, the adiabatic transport of the density matrix is
\begin{eqnarray}
\rho_{strong-ad}(s) & = & \rho_{00}(s) \\
& = & |\zeta_0(s)\rangle\langle \zeta_0(s)| \nonumber \\
& & +\frac{\epsilon V_1}{-\hbar r(s) + \epsilon(V_0 - V_1 \cos \varphi(s) \sin \theta(s))} \nonumber \\
& & \times\left( (e^{2\imath \varphi(s)} \sin^2 \frac{\theta(s)}{2} - \cos^2 \frac{\theta(s)}{2}) |\zeta_0(s)\rangle \langle \zeta_1(s)| \right.\nonumber \\
& & \left. +(e^{-2\imath \varphi(s)} \sin^2 \frac{\theta(s)}{2} - \cos^2 \frac{\theta(s)}{2}) |\zeta_1(s)\rangle \langle \zeta_0(s)| \right)
\end{eqnarray}
and if the dynamics is weakly adiabatic, the adiabatic transport of the density matrix is
\begin{equation}
\rho_{weak-ad}(s) = \Ad\left[\Teg^{-\ihbar^{-1} T \int_0^s E_0^{(1)}(\sigma)d\sigma} \Ted^{- \int_0^s A^{(1)}_0 (\sigma)d\sigma} \right] \rho_{00}(s)
\end{equation}
with
\begin{equation}
E_0^{(1)} = \lambda_{00} |\zeta_{00}^{(1)}\rangle \langle \zeta_{00}^{(1)}|+\lambda_{10} |\zeta_{10}^{(1)}\rangle \langle \zeta_{10}^{(1)}|
\end{equation}
($\eta^{(1)} = 0$ because $(\xi_\beta)$ are independent of $s$), and
\begin{equation}
A_0^{(1)} = \sum_{b,c=0}^{1}\langle \zeta_{b0}^{(1)}|\zeta_{c0}^{(1)\prime}\rangle |\zeta_{b0}^{(1)}\rangle \langle \zeta_{c0}^{(1)}| 
\end{equation}
where
\begin{eqnarray}
\lambda_{00} & = & \frac{\hbar}{2} r (\cos \theta-1) + \epsilon (V_0 - V_1 \cos \varphi \sin \theta) + \mathcal O(\epsilon^2) \\
\lambda_{10} & = & \frac{\hbar}{2} r (\cos \theta+1) + \epsilon (V_0 + V_1 \cos \varphi \sin \theta) + \mathcal O(\epsilon^2) \\
\lambda_{01} & = & \frac{\hbar}{2} r (\cos \theta-1) + \hbar\omega_e + \epsilon (2V_0 - V_2 \cos \varphi \sin \theta) + \mathcal O(\epsilon^2) \\
\lambda_{11} & = & \frac{\hbar}{2} r (\cos \theta+1) + \hbar\omega_e + \epsilon (2V_0 + V_2 \cos \varphi \sin \theta) + \mathcal O(\epsilon^2) \\
\end{eqnarray}
and
\begin{eqnarray}
\zeta_{00}^{(1)} & = & \zeta_0 + \frac{\epsilon V_1 (e^{-2\imath \varphi} \sin^2 \frac{\theta}{2} - \cos^2 \frac{\theta}{2})}{-\hbar r+\epsilon(V_0-V_1 \cos \varphi \sin \theta)} \zeta_1 \\
\zeta_{10}^{(1)} & = & \zeta_1 + \frac{\epsilon V_1 (e^{2\imath \varphi} \sin^2 \frac{\theta}{2} - \cos^2 \frac{\theta}{2})}{\hbar r+\epsilon(V_0+V_1 \cos \varphi \sin \theta)} \zeta_0 \\
 \zeta_{01}^{(1)} & = & \zeta_0 + \frac{\epsilon V_2 (e^{-2\imath \varphi} \sin^2 \frac{\theta}{2} - \cos^2 \frac{\theta}{2})}{-\hbar r+\epsilon(2V_0-V_1 \cos \varphi \sin \theta)} \zeta_1 \\
\zeta_{11}^{(1)} & = & \zeta_1 + \frac{\epsilon V_2 (e^{2\imath \varphi} \sin^2 \frac{\theta}{2} - \cos^2 \frac{\theta}{2})}{\hbar r+\epsilon(2V_0+V_1 \cos \varphi \sin \theta)} \zeta_0
\end{eqnarray}
The energies of the both qubits are represented figure \ref{lambda}.
\begin{figure}
\begin{center}
\includegraphics[width=9cm]{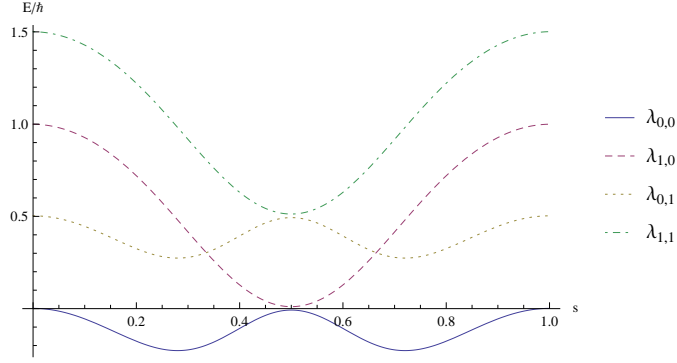}
\caption{\label{lambda} Instantaneous energies of the both atoms during the control with respect to the reduced time (with $\hbar \omega_e = 0.5\ au$, $r_{max}=1\ au$, $r_{min} = 0.02\ au$, $V_0=3\ au$, $V_1=1.5\ au$, $V_2=0.5\ au$, $V_3=2.5\ au$, $\theta_{max}=\frac{\pi}{2}$ and $\epsilon = 5\times10^{-4}$ ( $au$: atomic unit)).}
\end{center}
\end{figure}

\subsubsection{Strong adiabatic regime:}
We study a strong adiabatic regime where $T = 20000\ au$, $\tau_\S = \inf_{s\in[0,1]}\frac{\hbar}{|\mu_1(s)-\mu_0(s)|} = 2\ au$ and $\theta^\epsilon = \frac{\hbar}{\epsilon \|V\|} = 21\ au$ ($au$: atomic unit). We have $T \gg \theta^\epsilon \sim \tau_\S$ and there do not have resonance between transitions of $\S$ and $\E$ involving $\phi_{00}$ as shown figure \ref{lambda}. The assumptions of the theorem \ref{strongth} are then satisfied. The population of the qubit state 0 $\langle \zeta_0(0)|\rho(s)|\zeta_0(0)\rangle$ and the coherence of the controlled atom $|\langle \zeta_0(0)|\rho(s)|\zeta_1(0)\rangle|$ (note that $(\zeta_0(0),\zeta_1(0))$ is the eigenstate of the bare atom $\S$ since the laser is off at $s=0$) are represented figure \ref{figstrongad}.
\begin{figure}
\begin{center}
\includegraphics[width=9cm]{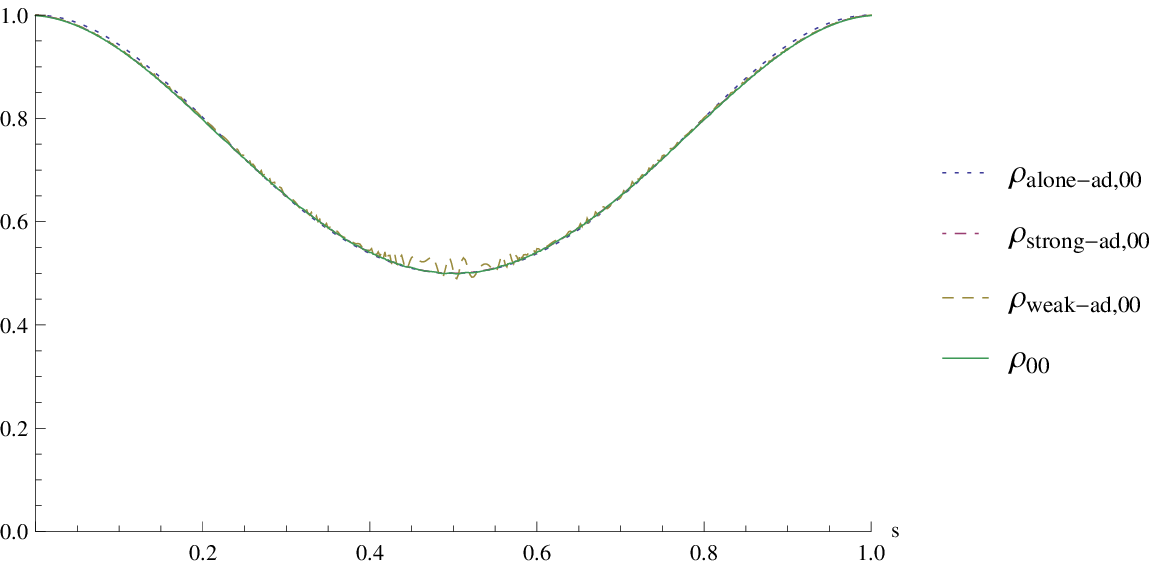} \includegraphics[width=9cm]{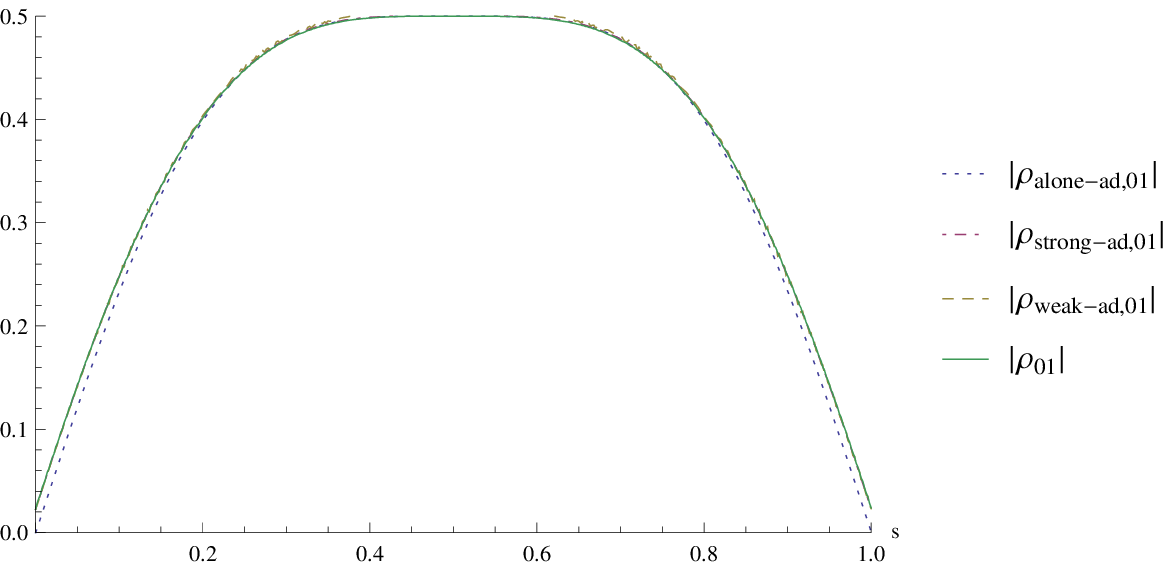}
\caption{\label{figstrongad} Population of the qubit state 0 $\rho_{\bullet,00} = \langle \zeta_0(0)|\rho_\bullet(s)|\zeta_0(0)\rangle$ (up) and coherence $\rho_{\bullet,01} = |\langle \zeta_0(0)|\rho_\bullet(s)|\zeta_1(0)\rangle|$ (down) for the exact dynamics ($\bullet=\varnothing$), the adiabatic transport formula with $\S$ alone ($\bullet = \text{alone-ad}$), the strong adiabatic transport formula ($\bullet=\text{strong-ad}$) and the weak adiabatic transport formula ($\bullet=\text{weak-ad}$); in conditions corresponding to a strong adiabatic regime (with $\hbar \omega_e = 1.5\ au$, $r_{max}=1\ au$, $r_{min} = 0.5\ au$, $V_0=3\ au$, $V_1=1.5\ au$, $V_2=0.5\ au$, $V_3=2.5\ au$, $\theta_{max}=\frac{\pi}{2}$ and $\epsilon = 1.6\times10^{-2}$ ($au$: atomic unit)).}
\end{center}
\end{figure}
The errors between the different adiabatic transport formulae and the exact dynamics are drawn figure \ref{errstrongad}.
\begin{figure}
\begin{center}
\includegraphics[width=10cm]{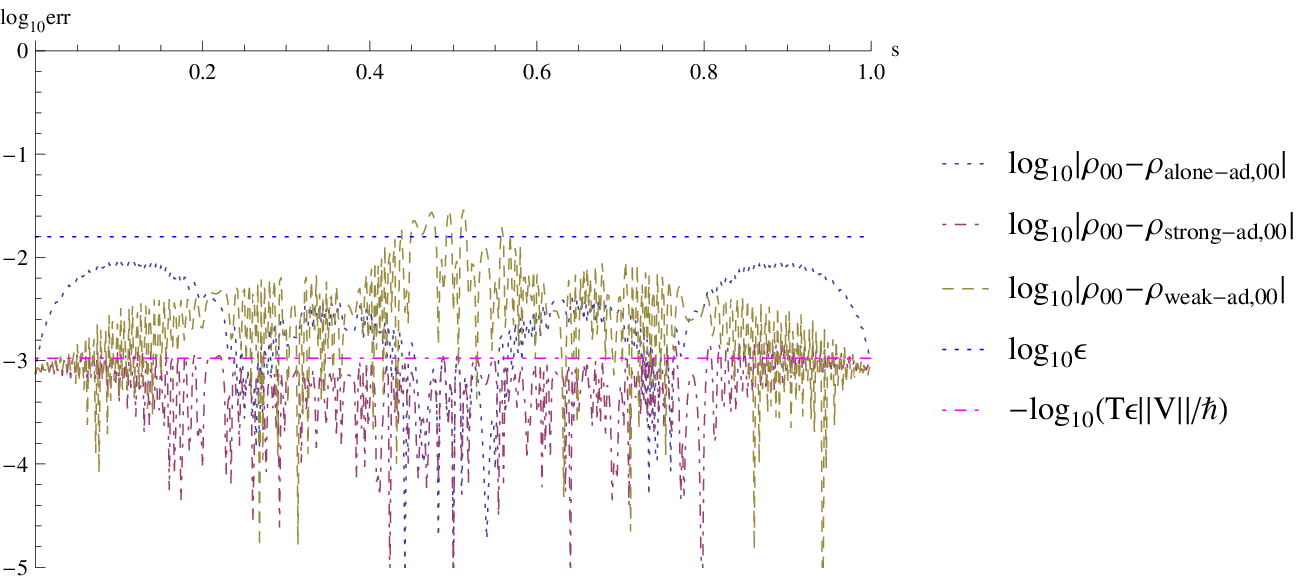} \includegraphics[width=10cm]{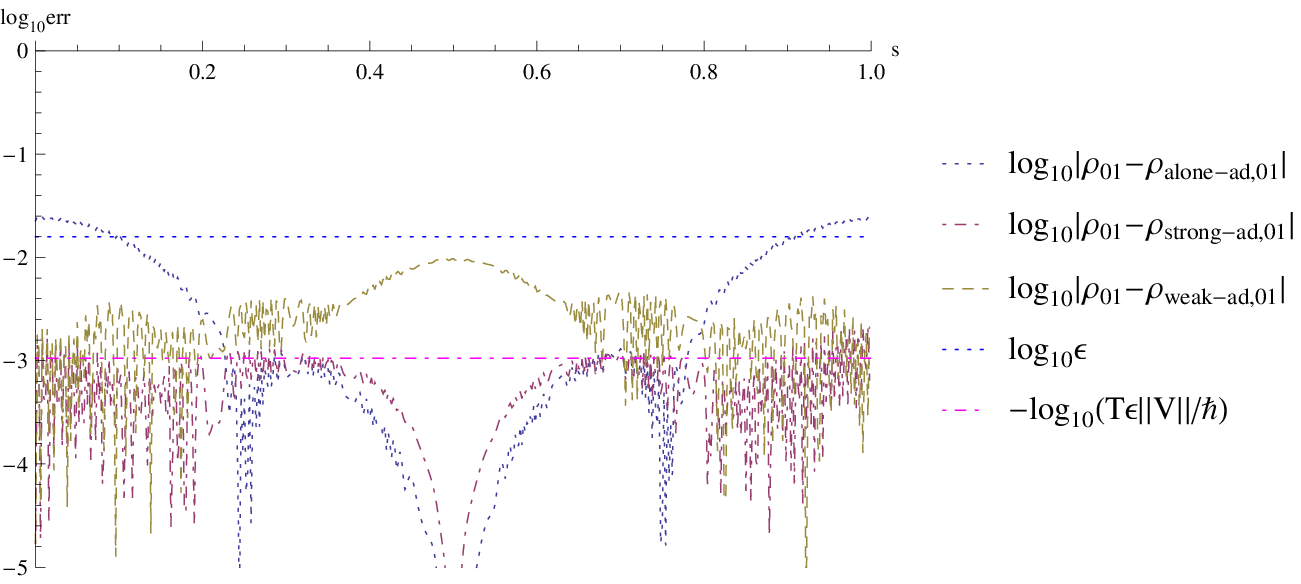}
\caption{\label{errstrongad} Errors in logarithmic scale between the approximations of the adiabatic transport formulae and the exact dynamics for the population of the qubit state 0 (up) and the coherence (down)  in conditions corresponding to a strong adiabatic regime (with $\hbar \omega_e = 1.5\ au$, $r_{max}=1\ au$, $r_{min} = 0.5\ au$, $V_0=3\ au$, $V_1=1.5\ au$, $V_2=0.5\ au$, $V_3=2.5\ au$, $\theta_{max}=\frac{\pi}{2}$ and $\epsilon = 1.6\times10^{-2}$ ($au$: atomic unit)).}
\end{center}
\end{figure}
The errors concerning the population reach $10^{-2}$ with the prediction of adiabatic transport formula for $\S$ alone (in accordance with the fact that the order of the coupling between $\S$ and $\E$ is $\epsilon = 1.6\times 10^{-2}$) while the errors concerning the coherence reach $2.5\times 10^{-2}$. The strong adiabatic transport formula permits to gain more than one order of magnitude on the errors in accordance with the theoretical error $\frac{\theta^\epsilon}{T} = 10^{-3}$.

\subsubsection{Weak adiabatic regime:}
We study a weak adiabatic regime where $T = 200\ au$, $\tau_\S = \inf_{s\in[0,1]}\frac{\hbar}{|\mu_1(s)-\mu_0(s)|} = 50\ au$, $\theta^\epsilon = \frac{\hbar}{\epsilon \|V\|} = 667\ au$ and $\tau_\E = \frac{1}{\omega_e}=2\ au$ ($au$: atomic unit). We have $\theta^\epsilon \sim \tau_\S \sim T$ and $T \gg \tau_\E$ and there are no quasi-resonance between transitions of $\S$ and $\E$. The assumptions of the theorem \ref{weakth} are then satisfied. The population of the qubit state 0 $\langle \zeta_0(0)|\rho(s)|\zeta_0(0)\rangle$ and the coherence of the controlled atom $|\langle \zeta_0(0)|\rho(s)|\zeta_1(0)\rangle|$ are represented figure \ref{figweakad}.
\begin{figure}
\begin{center}
\includegraphics[width=9cm]{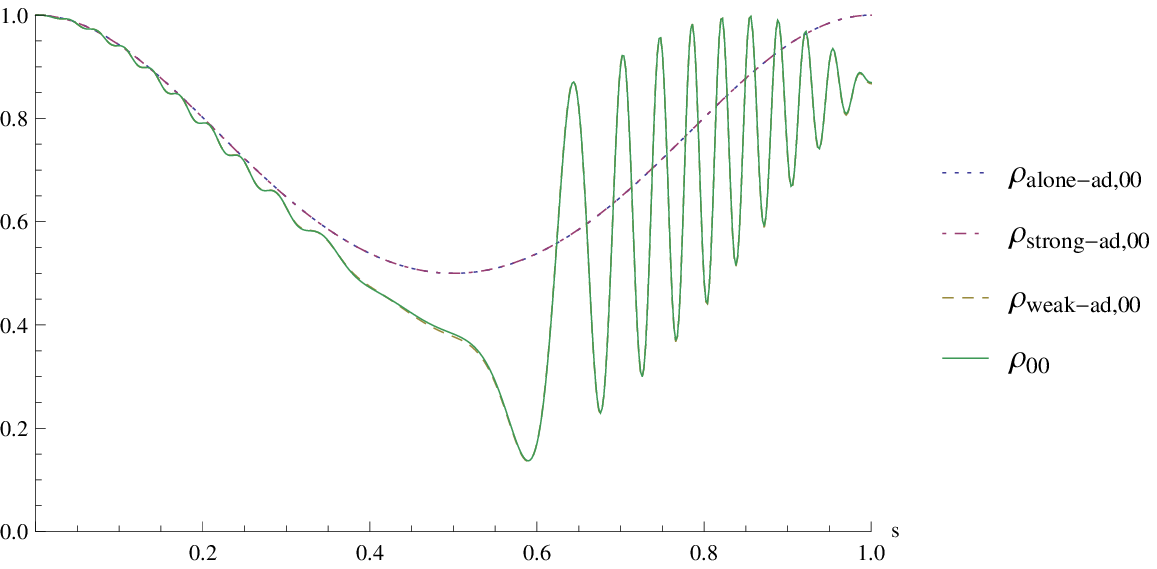} \includegraphics[width=9cm]{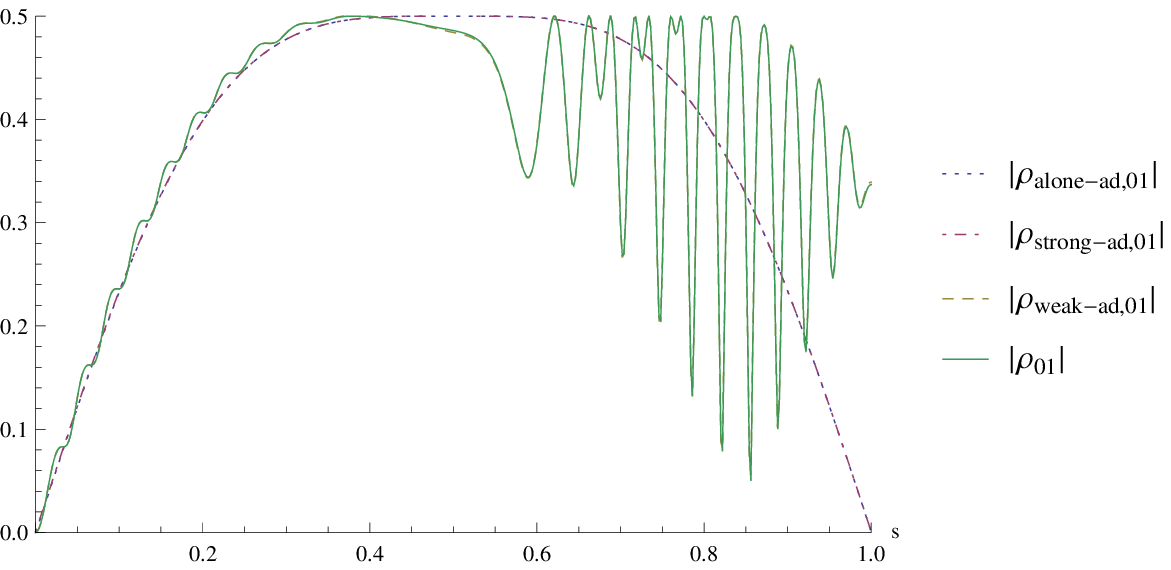}
\caption{\label{figweakad} Population of the qubit state 0 $\rho_{\bullet,00} = \langle \zeta_0(0)|\rho_\bullet(s)|\zeta_0(0)\rangle$ (up) and coherence $\rho_{\bullet,01} = |\langle \zeta_0(0)|\rho_\bullet(s)|\zeta_1(0)\rangle|$ (down) for the exact dynamics ($\bullet=\varnothing$), the adiabatic transport formula with $\S$ alone ($\bullet = \text{alone-ad}$), the strong adiabatic transport formula ($\bullet=\text{strong-ad}$) and the weak adiabatic transport formula ($\bullet=\text{weak-ad}$); in conditions corresponding to a weak adiabatic regime (with $\hbar \omega_e = 0.5\ au$, $r_{max}=1\ au$, $r_{min} = 0.02\ au$, $V_0=3\ au$, $V_1=1.5\ au$, $V_2=0.5\ au$, $V_3=2.5\ au$, $\theta_{max}=\frac{\pi}{2}$ and $\epsilon = 5\times10^{-4}$ ($au$: atomic unit)). Remark: the alone and the strongly adiabatic cases are graphically merged; the weak adiabatic and the exact cases are graphically merged.}
\end{center}
\end{figure}
The errors between the different adiabatic transport formulae and the exact dynamics are drawn figure \ref{errweakad}.
\begin{figure}
\begin{center}
\includegraphics[width=10cm]{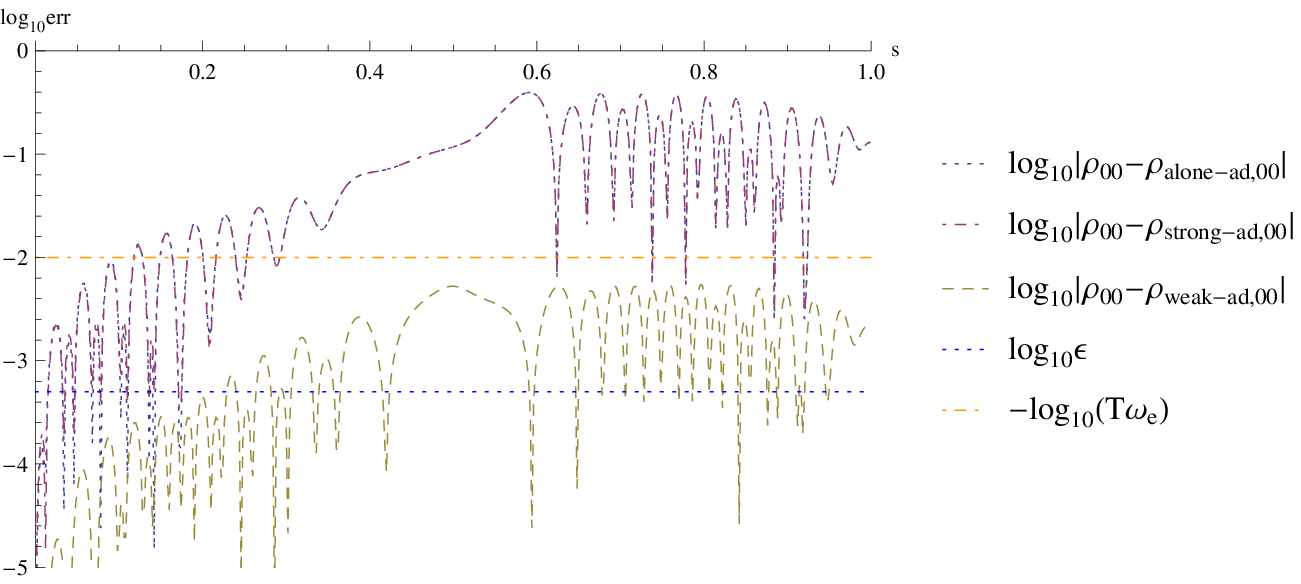} \includegraphics[width=10cm]{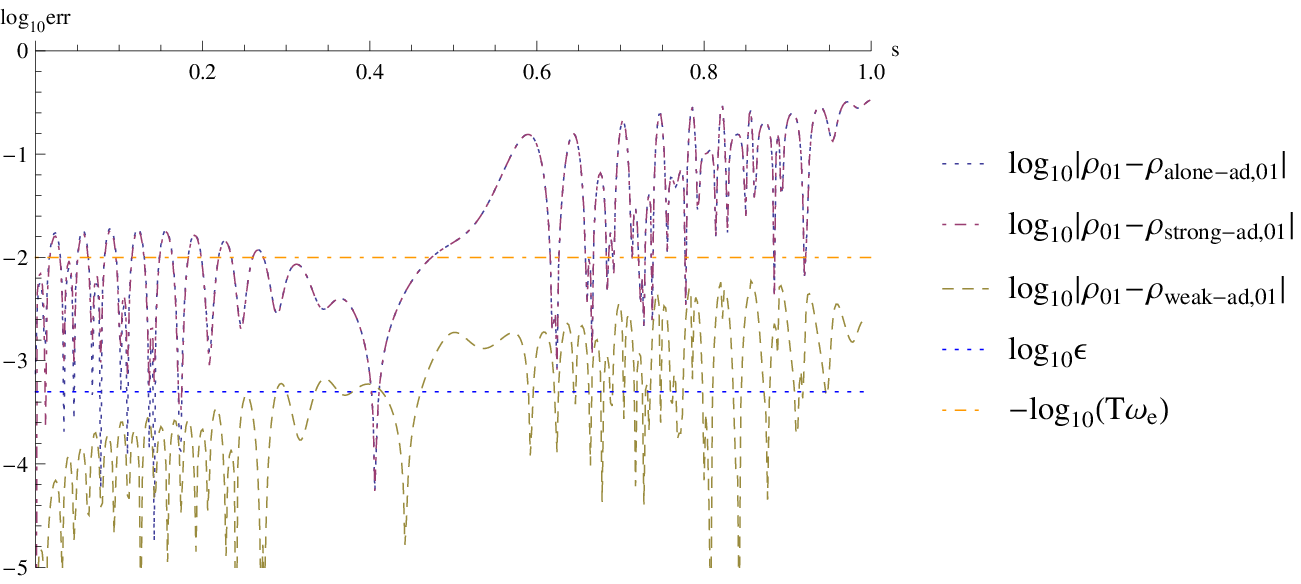}
\caption{\label{errweakad} Errors in logarithmic scale between the approximations of the adiabatic transport formulae and the exact dynamics for the population of the qubit state 0 (up) and the coherence (down)  in conditions corresponding to a weak adiabatic regime (with $\hbar \omega_e = 0.5\ au$, $r_{max}=1\ au$, $r_{min} = 0.02\ au$, $V_0=3\ au$, $V_1=1.5\ au$, $V_2=0.5\ au$, $V_3=2.5\ au$, $\theta_{max}=\frac{\pi}{2}$ and $\epsilon = 5\times10^{-4}$ ($au$: atomic unit)). Remark: the alone and the strongly adiabatic cases are graphically merged.}
\end{center}
\end{figure}
The errors of the prediction of the adiabatic transport formula with $\S$ alone is now very large in accordance with the very small gap between the two eigenvalues of $H_\S(s)$ during the dynamics. The weak adiabatic transport formula provides a very good approximation with an error smaller than $\frac{\tau_\E}{T} = 10^{-3}$ in accordance with the theoretical error $\max(\frac{\tau_\E}{T},\epsilon^2)$.

\subsection{Control of a spin in the middle of a chain}
\subsubsection{The model:}
We consider a Heisenberg line chain of $2N+1$ spins with nearest neighbour interaction. A constant and uniform magnetic field $\vec B_{Zeeman} = -\frac{\omega_e}{2} \vec e_z$ is applied on all the spins of the chain in order to split the energy levels of the spins by a Zeeman effect. A time dependent magnetic field $\vec B_{control}(s)$ is applied only on the middle spin denoted by $\S$ to control it. $\S$ is governed by the Hamiltonian
\begin{eqnarray}
H_\S(s) & = & \vec B(s) \cdot \vec S \\
& = & \frac{\hbar}{2}(B_x(s) \sigma_x + B_y(s) \sigma_y + B_z(s) \sigma_z) 
\end{eqnarray}
$\vec S = \frac{\hbar}{2}(\sigma_x,\sigma_y,\sigma_z)$ is the spin operator ($\{\sigma_i\}_i$ are the Pauli matrices) and $\vec B(s) = \vec B_{control}(s) + \vec B_{Zeeman}$. The rest of the chain is denoted by $\E$ and is described by the Hilbert space $\mathcal H_\E = \mathcal H_{\E l} \otimes \mathcal H_{\E r}$ where $\mathcal H_{\E l/r} = (\mathbb C^2)^{\otimes N}$ are the Hilbert spaces of the half chains on the left and on the right of the controlled spin. $\E$ is governed by the Hamiltonian
\begin{equation}
H_\E = H_{\mathcal C} \otimes \id^{\otimes N} + \id^{\otimes N} \otimes H_{\mathcal C}
\end{equation}
\begin{eqnarray}
H_{\mathcal C} & = & \sum_{n=1}^{N} \id^{\otimes(n-1)} \otimes \vec B_{Zeeman} \cdot \vec S \otimes \id^{\otimes(N-n)} \nonumber \\
& & \quad - J \sum_{n=1}^{N-1} \id^{\otimes(n-1)} \otimes \vec S \odot \vec S \otimes \id^{\otimes(N-n-1)}
\end{eqnarray}
where $\id$ denotes the identity operator for one spin, $\vec S \odot \vec S = \sum_{i=x,y,z} S_i \otimes S_i$, and $J$ is the coupling constant. The interaction between $\S$ and $\E$ is described by
\begin{equation}
V_{\S-\E} = -J \vec S \odot \left(\id^{\otimes (N-1)} \otimes \vec S \otimes \id^{\otimes N} + \id^{\otimes N} \otimes \vec S \otimes \id^{\otimes(N-1)} \right)
\end{equation}
with $V_{\S-\E} \in \mathcal H_\S \otimes \mathcal H_\E$. The coupling constant $J \ll 1$ is the perturbative parameter.\\
Let $B(s) = \|\vec B(s)\|$, $\theta(s) = \arccos \frac{B_z(s)}{B(s)}$ and $\varphi(s) = \arctan \frac{B_y(s)}{B_x(s)}$. The eigenvalues and the eigenvectors of $\S$ are
\begin{eqnarray}
\mu_0(s) = - \frac{\hbar}{2} B(s) & \qquad & \zeta_0(s) = \left(\begin{array}{c} - \sin \frac{\theta(s)}{2} \\ e^{\imath \varphi(s)} \cos \frac{\theta(s)}{2} \end{array} \right) \\
\mu_1(s) = + \frac{\hbar}{2} B(s) & \qquad & \zeta_1(s) = \left(\begin{array}{c} e^{-\imath \varphi(s)} \cos \frac{\theta(s)}{2} \\ \sin \frac{\theta(s)}{2} \end{array} \right)
\end{eqnarray}
The eigenvalues of $\E_l$ or $\E_r$ for $N=3$ are
\begin{eqnarray}
\nu_{(000)} & = & -3 \frac{\hbar \omega_e}{4} - J \frac{\hbar^2}{2} + \mathcal O(J^2) \\
\nu_{(100)-(001)} & = & - \frac{\hbar \omega_e}{4} + \mathcal O(J^2)  \\
\nu_{(001)+(010)+(100)} & = & - \frac{\hbar \omega_e}{4} - J \frac{\hbar^2}{2} + \mathcal O(J^2) \\
\nu_{(001)-2(010)+(100)} & = & - \frac{\hbar \omega_e}{4} + J\hbar^2 + \mathcal O(J^2) \\
\nu_{(110)-(011)} & = & \frac{\hbar \omega_e}{4} + \mathcal O(J^2)  \\
\nu_{(011)+(101)+(110)} & = & \frac{\hbar \omega_e}{4} - J \frac{\hbar^2}{2} + \mathcal O(J^2) \\
\nu_{(011)-2(101)+(110)} & = & \frac{\hbar \omega_e}{4} + J\hbar^2 + \mathcal O(J^2) \\
\nu_{(111)} & = & 3 \frac{\hbar \omega_e}{4} - J \frac{\hbar^2}{2} + \mathcal O(J^2)
\end{eqnarray}
which are associated with the eigenvectors
\begin{eqnarray}
\xi_{(ijk)} & = & |i\rangle \otimes |j\rangle \otimes |k\rangle \\
\xi_{a(ijk)+b(lmn)+c(opq)} & = & \frac{1}{\sqrt{a^2+b^2+c^2}}(a \xi_{ijk} + b \xi_{lmn} + c \xi_{opq})
\end{eqnarray}
$(|i\rangle)_{i=0,1})$ being the eigenstates of an isolated spin.\\
The control is fixed by the following variation of the control parameters:
\begin{eqnarray}
B(s) & = & B_0(1-e^{-(s-0.5)^2/\Delta s^2})+B_{min} \\
\theta(s) & = & \pi(1-\sin(\pi s)) \\
\varphi(s) & = & 2\pi s
\end{eqnarray}

\subsubsection{Adiabatic transports:}
We start with the chain in a state $\phi_{0\alpha_l\alpha_d}$ where $\alpha_l,\alpha_g \in \{(000),(100)-(001),...,(011)-2(101)+(110),(111)\}$ corresponding to the states of left and right half chains. The adiabatic transport of the density matrix for $\S$ alone is
\begin{equation}
\rho_{alone-ad}(s) = |\zeta_0(s)\rangle \langle \zeta_0(s)|
\end{equation}
If the dynamics of the chain is strongly adiabatic, the adiabatic transport of the density matrix is
\begin{eqnarray}
\rho_{strong-ad}(s) & = & \rho_{0\alpha_l \alpha_r}(s) \\
& = & |\zeta_0(s)\rangle \langle \zeta_0(s)| \nonumber \\
& & \quad + \frac{J\hbar}{4} n_{\alpha_l\alpha_d} \frac{\sin \theta(s)}{-B(s)+\frac{J\hbar}{4} n_{\alpha_l\alpha_r} \cos \theta(s)} \nonumber \\
& & \quad \times \left(|\zeta_1(s)\rangle \langle \zeta_0(s)| + |\zeta_0(s)\rangle \langle \zeta_1(s)| \right)
\end{eqnarray}
where $n_{\alpha_g \alpha_d} = n_{\alpha_g} + n_{\alpha_d}$ is a number defined by table \ref{alpha}.
\begin{table}
\begin{center}
\caption{\label{alpha} Values of the number $n_\alpha$ characterizing the coupling in a half chain in the state $\xi_\alpha$.}
\begin{tabular}{r|c}
$\alpha$ & $n_\alpha$ \\
\hline
$(111)$ & $-1$ \\
$(110)-2(101)+(011)$ & $-\frac{2}{3}$ \\
$(110)+(101)+(011)$ & $-\frac{1}{3}$ \\
$(110)-(011)$ & $0$ \\
$(100)-(001)$ & $0$ \\
$(100)+(010)+(001)$ & $\frac{1}{3}$ \\
$(100)-2(010)+(001)$ & $\frac{2}{3}$ \\
$(000)$ & $1$
\end{tabular}
\end{center}
\end{table}

If the dynamics is weakly adiabatic, the adiabatic transport of the density matrix is
\begin{eqnarray}
& & \rho_{weak-ad}(s) \nonumber \\
& & \quad = \Ad\left[\Teg^{-\ihbar^{-1} T \int_0^s E_{\alpha_l\alpha_r}^{(1)}(\sigma)d\sigma} \Ted^{- \int_0^s A^{(1)}_{\alpha_l\alpha_r} (\sigma)d\sigma} \right] \rho_{0\alpha_l \alpha_r}(s)
\end{eqnarray}
with
\begin{equation}
E_{\alpha_l\alpha_r}^{(1)} = \lambda_{0\alpha_l\alpha_r} |\zeta_{0\alpha_l\alpha_r}^{(1)}\rangle \langle \zeta_{0\alpha_l\alpha_r}^{(1)}|+\lambda_{1\alpha_l\alpha_r} |\zeta_{1\alpha_l\alpha_r}^{(1)}\rangle \langle \zeta_{1\alpha_l\alpha_r}^{(1)}|
\end{equation}
($\eta^{(1)} = 0$ because $(\xi_\beta)$ are independent of $s$), and
\begin{equation}
A_{\alpha_l\alpha_r}^{(1)} = \sum_{b,c=0}^{1}\langle \zeta_{b\alpha_l\alpha_r}^{(1)}|\zeta_{c\alpha_l\alpha_r}^{(1)\prime}\rangle |\zeta_{b\alpha_l\alpha_r}^{(1)}\rangle \langle \zeta_{c\alpha_l\alpha_r}^{(1)}| 
\end{equation}
where
\begin{eqnarray}
\lambda_{0\alpha_l\alpha_r} & = & - \frac{\hbar}{2} B + \nu_{\alpha_l}+\nu_{\alpha_r} + \frac{J\hbar^2}{4}n_{\alpha_l \alpha_r} \cos \theta + \mathcal O(J^2)  \\
\lambda_{1\alpha_l\alpha_r} & = & \frac{\hbar}{2} B + \nu_{\alpha_l}+\nu_{\alpha_r} - \frac{J\hbar^2}{4}n_{\alpha_l \alpha_r} \cos \theta + \mathcal O(J^2)
\end{eqnarray}
and
\begin{eqnarray}
\zeta_{0\alpha_l\alpha_r}^{(1)} & = & \zeta_0 -\frac{J\hbar}{4}n_{\alpha_l\alpha_r}\frac{\sin \theta}{B-\frac{J\hbar}{4}n_{\alpha_l\alpha_r}\cos \theta} \zeta_1  \\
\zeta_{1\alpha_l\alpha_r}^{(1)} & = & \zeta_1 +\frac{J\hbar}{4}n_{\alpha_l\alpha_r}\frac{\sin \theta}{B-\frac{J\hbar}{4}n_{\alpha_l\alpha_r}\cos \theta} \zeta_0 
\end{eqnarray}

\subsubsection{Strong adiabatic regime:}
We study a strong adiabatic regime where $T = 5\times 10^3\ au$, $\tau_\S = \inf_{s\in[0,1]}\frac{\hbar}{|\mu_1(s)-\mu_0(s)|} = 1.5\ au$ and $\theta^J = \frac{\hbar}{\|V_{\S-\E}\|} = 2\times 10^2\ au$ ($au$: atomic unit). We have $T \gg \theta^J \gg \tau_\S$ assuring that the assumptions of theorem \ref{strongth} are satisfied. The population of the spin state 0 $\langle \zeta_0(0)|\rho(s)|\zeta_0(0)\rangle$ and the coherence of the controlled spin $|\langle \zeta_0(0)|\rho(s)|\zeta_1(0)\rangle|$ (note that $(\zeta_0(0),\zeta_1(0))$ is the eigenstate of the ``free'' spin $\S$ since the magnetic field of control is off at $s=0$) are represented figure \ref{figstrongchain}.
\begin{figure}
\begin{center}
\includegraphics[width=9cm]{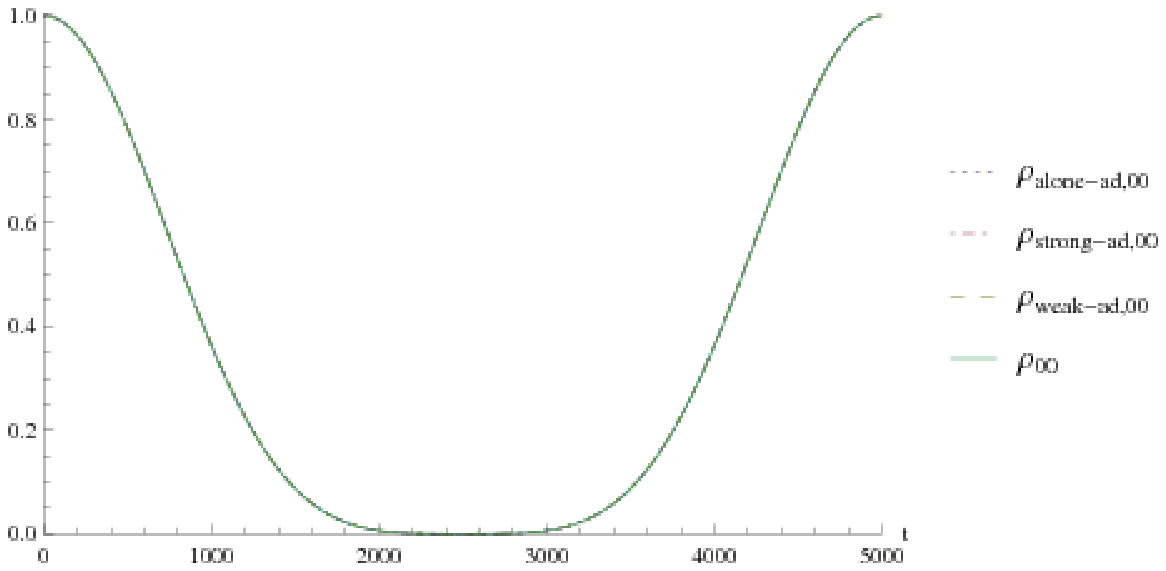} \includegraphics[width=9cm]{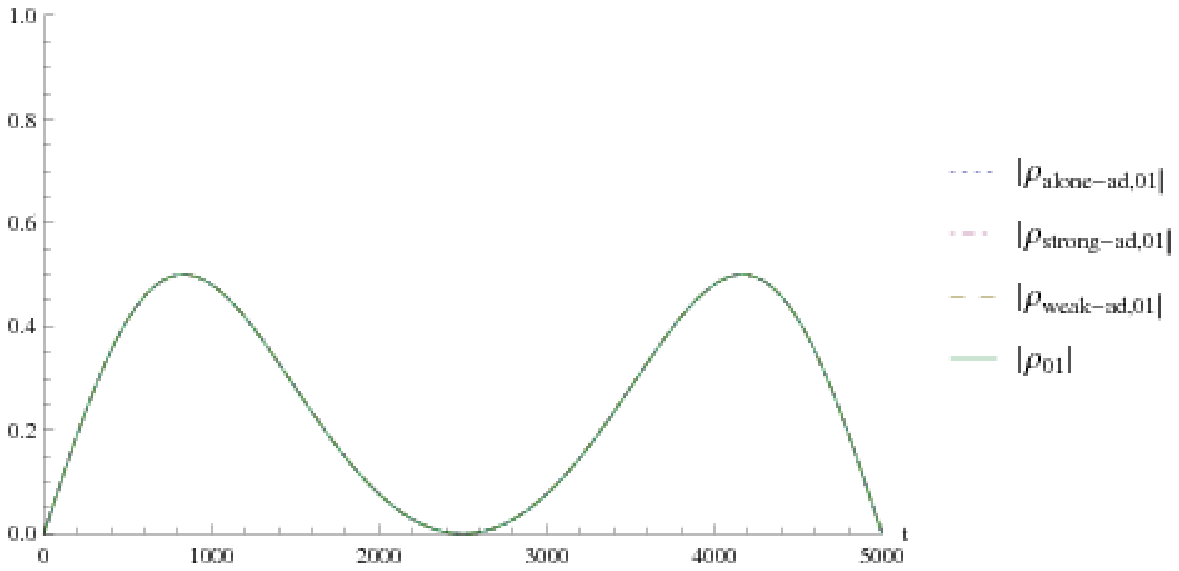}
\caption{\label{figstrongchain} Population of the spin state 0 $\rho_{\bullet,00} = \langle \zeta_0(0)|\rho_\bullet(s)|\zeta_0(0)\rangle$ (up) and coherence $\rho_{\bullet,01} = |\langle \zeta_0(0)|\rho_\bullet(s)|\zeta_1(0)\rangle|$ (down) for the exact dynamics ($\bullet=\varnothing$), the adiabatic transport formula with $\S$ alone ($\bullet = \text{alone-ad}$), the strong adiabatic transport formula ($\bullet=\text{strong-ad}$) and the weak adiabatic transport formula ($\bullet=\text{weak-ad}$); in conditions corresponding to a strong adiabatic regime (with $\hbar \omega_e = 2\ au$, $B_0=1\ au$, $B_{min} = 0.67\ au$, and  $J= 2\times 10^{-2}\ au$ ($au$: atomic unit)).}
\end{center}
\end{figure}
The errors between the different adiabatic transport formulae and the exact dynamics are drawn figure \ref{errstrongchain}.
\begin{figure}
\begin{center}
\includegraphics[width=10cm]{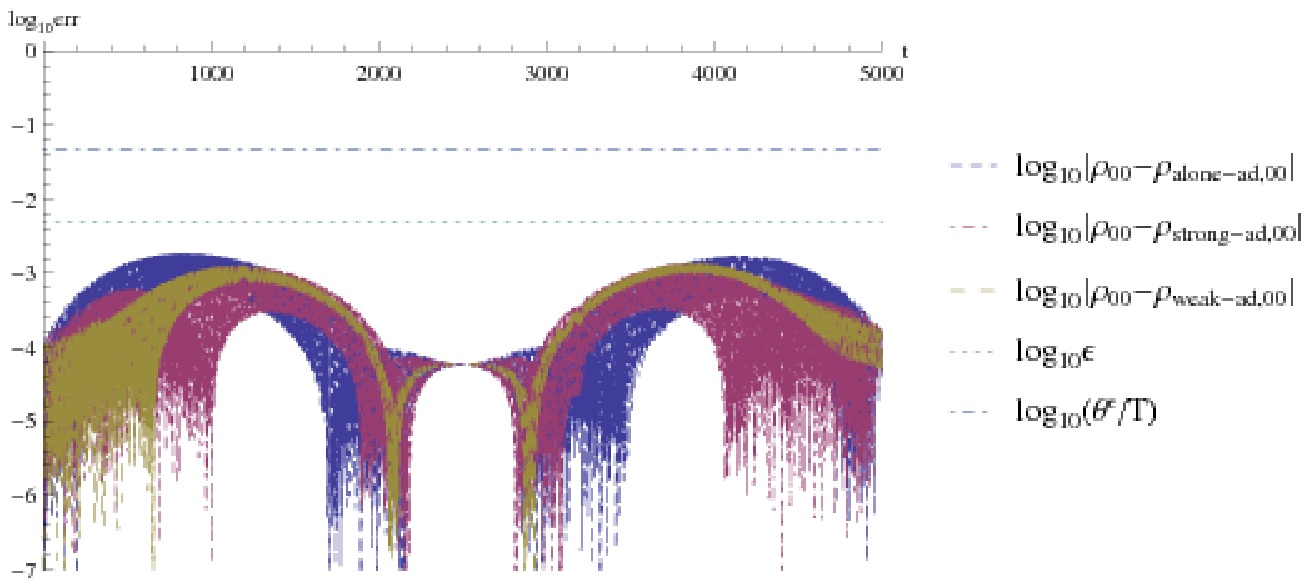} \includegraphics[width=10cm]{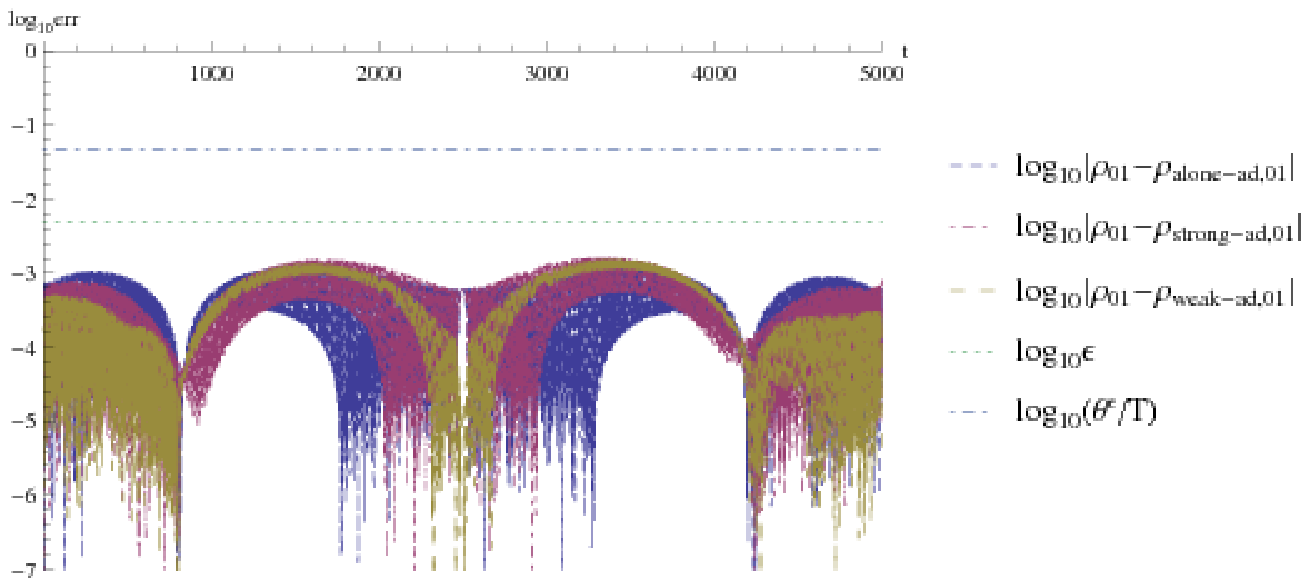}
\caption{\label{errstrongchain} Errors in logarithmic scale between the approximations of the adiabatic transport formulae and the exact dynamics for the population of the qubit state 0 (up) and the coherence (down)  in conditions corresponding to a strong adiabatic regime (with $\hbar \omega_e = 2\ au$, $B_0=1\ au$, $B_{min} = 0.67\ au$, and  $J= 2\times 10^{-2}\ au$ ($au$: atomic unit)).}
\end{center}
\end{figure}
A numerical study shows that a purely strong adiabatic regime seems not to be present for this system which presents rather regimes where the adiabatic approximation without environment, the strong adiabatic approximation and the weak adiabatic approximation are not clearly distinguishable. Nevertheless we see figure \ref{errstrongchain} that the strong adiabatic transport formula induces globally less errors.

\subsubsection{Weak adiabatic adiabatic regime:}
We study a weak adiabatic regime where $T = 50\ au$, $\tau_\S = \inf_{s\in[0,1]}\frac{\hbar}{|\mu_1(s)-\mu_0(s)|} = 10^2\ au$, $\theta^J = \frac{\hbar}{\|V_{\S-\E}\|} = 10^3\ au$ and $\tau_\E = \frac{1}{\omega_e} = 0.5\ au$ ($au$: atomic unit). We have $T \gg \tau_\E$ and $\theta^J \sim \tau_S \not\ll T$ assuring that the assumptions of theorem \ref{weakth} are satisfied. The population of the spin state 0 $\langle \zeta_0(0)|\rho(s)|\zeta_0(0)\rangle$ and the coherence of the controlled spin $|\langle \zeta_0(0)|\rho(s)|\zeta_1(0)\rangle|$ are represented figure \ref{figweakchain}.
\begin{figure}
\begin{center}
\includegraphics[width=9cm]{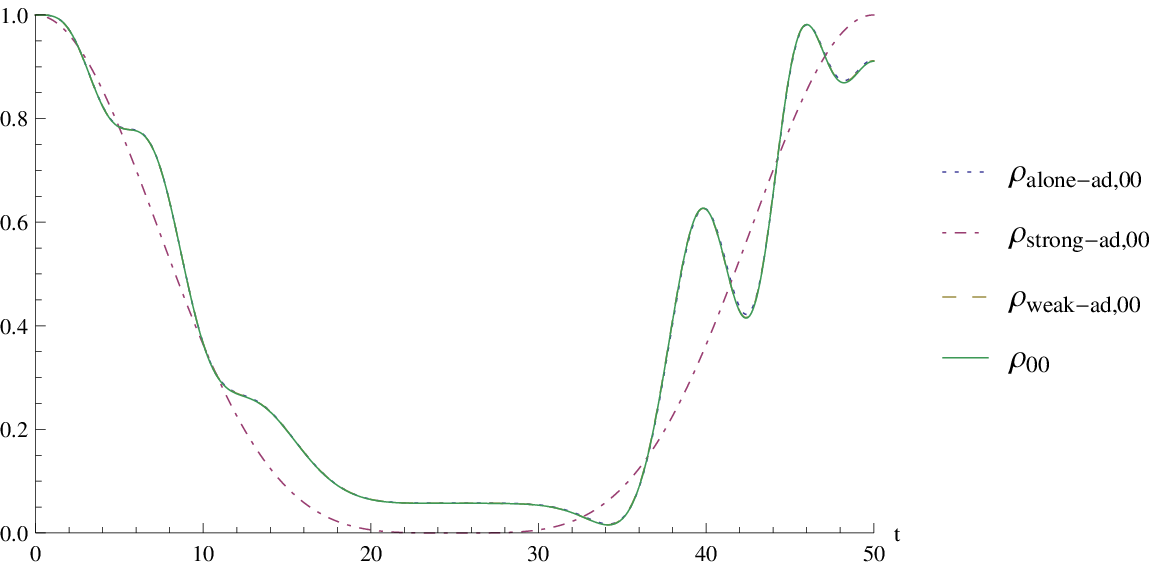} \includegraphics[width=9cm]{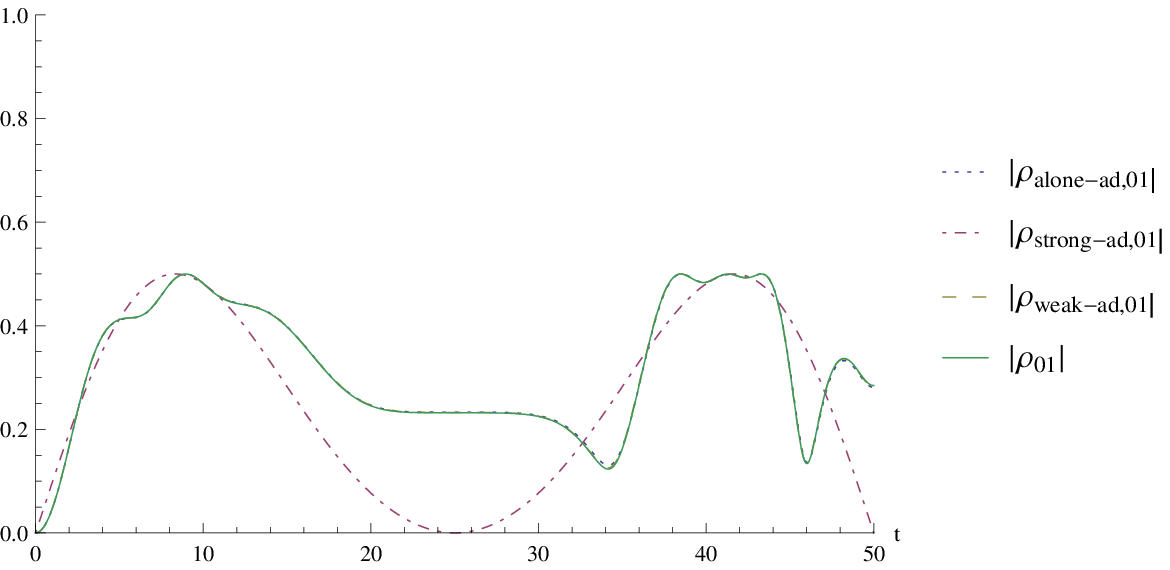}
\caption{\label{figweakchain} Population of the spin state 0 $\rho_{\bullet,00} = \langle \zeta_0(0)|\rho_\bullet(s)|\zeta_0(0)\rangle$ (up) and coherence $\rho_{\bullet,01} = |\langle \zeta_0(0)|\rho_\bullet(s)|\zeta_1(0)\rangle|$ (down) for the exact dynamics ($\bullet=\varnothing$), the adiabatic transport formula with $\S$ alone ($\bullet = \text{alone-ad}$), the strong adiabatic transport formula ($\bullet=\text{strong-ad}$) and the weak adiabatic transport formula ($\bullet=\text{weak-ad}$); in conditions corresponding to a weak adiabatic regime (with $\hbar \omega_e = 2\ au$, $B_0=1\ au$, $B_{min} = 10^{-2}\ au$, and $J = 2\times 10^{-3}$ ($au$: atomic unit)). Remark: the alone and the strongly adiabatic cases are graphically merged; the weak adiabatic and the exact cases are graphically merged.}
\end{center}
\end{figure}
The errors between the different adiabatic transport formulae and the exact dynamics are drawn figure \ref{errweakchain}.
\begin{figure}
\begin{center}
\includegraphics[width=10cm]{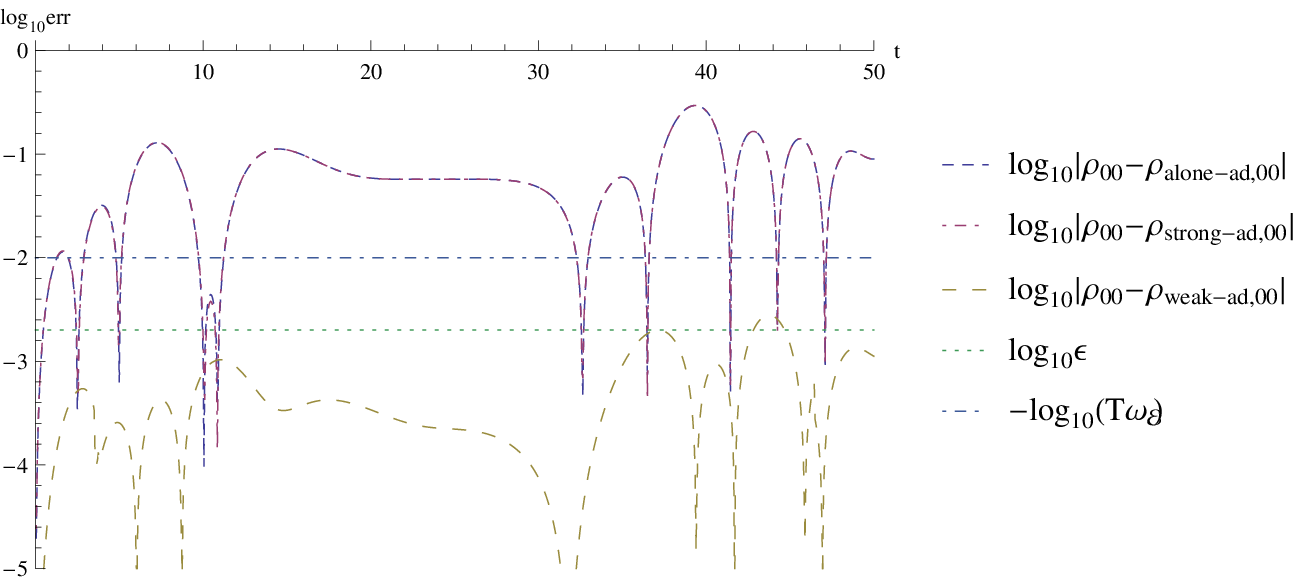} \includegraphics[width=10cm]{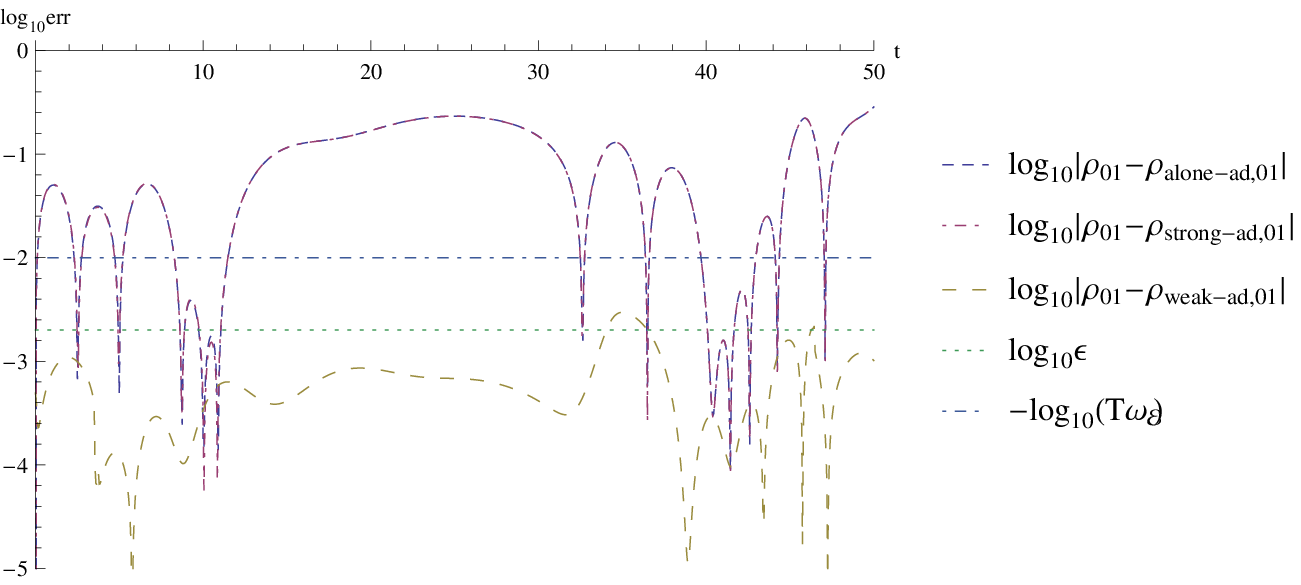}
\caption{\label{errweakchain} Errors in logarithmic scale between the approximations of the adiabatic transport formulae and the exact dynamics for the population of the spin state 0 (up) and the coherence (down)  in conditions corresponding to a weak adiabatic regime (with $\hbar \omega_e = 2\ au$, $B_0=1\ au$, $B_{min} = 10^{-2}\ au$, and $J = 2\times 10^{-3}$ ($au$: atomic unit)). Remark: the alone and the strongly adiabatic cases are graphically merged.}
\end{center}
\end{figure}
The errors of the prediction of the adiabatic transport formula with $\S$ alone is now very large in accordance with the very small gap between the two eigenvalues of $H_\S(s)$ during the dynamics. The weak adiabatic transport formula provides a very good approximation with an error smaller than $\frac{\tau_\E}{T} = 10^{-2}$ in accordance with theoretical error $\max(\frac{\tau_\E}{T},\epsilon^2)$.

\section{Conclusion}
We have shown that operator-valued geometric phases like defined by \cite{Sjoqvist,Tong,Dajka,Andersson,Viennot1,Viennot2} are exhibited by bipartite quantum systems in an adiabatic approximation with a perturbative coupling between the both parts of the system. This result remains valid if the bipartite system is constituted by a small subsystem and a large environment. Nevertheless, for a very large environment (a reservoir) the adiabatic theorem assumptions of no resonance or no quasi-resonance between transitions of $\S$ and $\E$ can be not satisfied since the spectrum of a reservoir is assimilated to a continuum \cite{Breuer}. These adiabatic operator-valued geometric phases arise when the evolution of the environment is strongly adiabatic (the favorable case for a quantum control of the subsystem) but with a subsystem evolution not necessarily adiabatic with respect to the control and to the environment effects. The operator-valued dynamical phase generator arising with the geometric phase generator, is a kind of effective Hamiltonian representing the system dressed by environment states. The second order adiabatic transport satisfies a kind of effective Lindblad equation.\\
The perturbative assumption restricts the field of applications of the present result to special situations. It would be interesting to prove that the adiabatic transport of density matrices exhibits also an operator-valued geometric phase with a strong interaction between the both parts of a bipartite system.

\appendix
\section{A corollary concerning the splitting of the time ordered exponential} \label{corap}
\begin{cor}
\label{cor}
Let $s \mapsto A(s) \in \mathcal L(\mathcal V)$ be a family of bounded anti-self-adjoint operators of an Hilbert space $\mathcal V$. Let $s \mapsto U_A(s) \in \mathcal U(\mathcal V)$ be the unitary operator strongly continuous with respect to $s$ and solution of the equation
\begin{equation}
 U_A' = -A U_A \qquad U_A(0) = 1_{\mathcal V}
\end{equation}
Let $s \mapsto B(s) \in \mathcal L(\mathcal V)$ be another family of bounded anti-self-adjoint operators, with the same notations we have
\begin{equation}
U_{A+B} = U_X U_A
\end{equation}
with 
\begin{equation}
X = A+B-U_X A U_X^{-1}
\end{equation}
\end{cor}

\begin{proof}
Let $X(s) \in \mathcal L(\mathcal V)$ be such that $U_X = U_{A+B} U_A^{-1}$.
\begin{eqnarray}
U_X U_A = U_{A+B} & \Rightarrow & U_X' U_A + U_X U_A' = U_{A+B}' \\
& \Rightarrow & -XU_X U_A - U_X A U_A = -(A+B) U_X U_A \\
& \Rightarrow & X = A+B - U_X A U_X^{-1}
\end{eqnarray}
\end{proof}
We note that $X$ is only implicitly defined.

\end{document}